\documentclass{scrartcl}

\def\fullversion{the appendix\xspace}

\usepackage{listings}
\usepackage{auxlib}
\usepackage{multirow}
\setlist[enumerate,1]{label=\arabic*), ref=\arabic*)}
\setlist[enumerate,2]{label=\alph*), ref=\alph{enumi}.\arabic*)}
\crefname{enumi}{Statement}{Statements}
\usepackage[letterpaper,margin=1in]{geometry}
\usepackage{cfr-lm}

\addtokomafont{disposition}{\rmfamily}

\makeatletter
\lst@AddToHook{OnEmptyLine}{\addtocounter{lstnumber}{-1}\vspace{-0.5\baselineskip}}%
\makeatother

\usepackage{framed}

\def\Vardef#1{%
	\expandafter\newcommand\csname #1\endcsname[1]{%
		\def\first{##1}%
		\def\second{*}%
		\ensuremath{\mathsf{\MakeLowercase #1}\ifx\first\second^{*}\else_{##1}\fi}%
	}%
}
\def\Vardefx#1#2{%
	\expandafter\newcommand\csname #1\endcsname[1]{%
		\ensuremath{\mathsf{#2}_{##1}}%
	}%
}
\def\T_#1{\ensuremath{T_{\text{\normalfont\textsc{\MakeLowercase{#1}}}}}}

\usepackage{float}
\floatname{algoflt}{Algorithm}
\newfloat{algoflt}{tbp}{loa}
\crefname{algoflt}{Algorithm}{Algorithms}

\usepackage[font=small,labelfont=bf]{caption}

\def\?#1{}
\newcommand\lcomment[2][t]{\color{gray}{\em\footnotesize$\triangleright$ #2}}

\newcommand\comment[2][t]{\color{gray}\hfill\hspace{0.5em}\begin{minipage}[#1]{0.3333\textwidth}{\em\footnotesize$\triangleright$ #2}\end{minipage}}
\def\dcmnumberstyle{}
\def\dcmbasicstyle{}
\makeatletter
\lstnewenvironment{lstalgo}[2][htb]{%
\lstset{%
	mathescape,%
	tabsize=4,%
	numbers=left,%
	numberstyle=\tiny\dcmnumberstyle,%
	numberblanklines=false,%
	basicstyle=\rmfamily\footnotesize\dcmbasicstyle,%
	columns=fullflexible,%
	escapechar=\#,%
	emph={%
		[1]if, else, then, and, or, return, exit, not, output %
		},
	emphstyle={%
		[1]\rmfamily\bfseries%
		},%
	}%
\setlength{\topsep}{0pt}%
\setlength{\parindent}{0pt}%
\setlength{\FrameSep}{0pt}%
\@skiphyperreftrue
\global\def\dcmnumberstyle{\addtocounter{lstnumber}{-1}\color{white}\global\def\dcmnumberstyle{\tiny}
}%

\algoflt[#1]\caption{#2}%
\framed\hspace*{2em}\minipage{\textwidth-3em}}{\@skiphyperreffalse\endminipage\endframed\endalgoflt}
\makeatother
\crefname{lst@lineno}{Line}{Lines}

\def\BigO#1{\smash{\LDAUOmicron{#1}}}
\def\LittleO#1{\smash{\LDAUomicron{#1}}}
\def\BigOmega#1{\smash{\LDAUOmega{#1}}}

\def\BigTheta#1{\smash{\LDAUTheta{#1}}}
\def\BigOTilde#1{\smash{\SOFTOmicron{#1}}}

\everydisplay\expandafter{\the\everydisplay%
\let\BigO\LDAUOmicron%
\let\LittleO\LDAUomicron%
\let\BigOmega\LDAUOmega%
\let\BigTheta\LDAUTheta%
\let\BigOTilde\SOFTOmicron%
}

\let\Probability\Prob

\let\epsilon\varepsilon

\newcommand{\Exp}[1]{\ensuremath{\exp\left(#1\right)}}

\newcommand\numberthis{\addtocounter{equation}{1}\tag{\theequation}}
\def\lefttag#1{\tag*{\makebox[0pt][l]{\hspace*{-\linewidth}#1}}}

\usepackage{xspace}
\def\ttrue{\ensuremath{\text{\normalfont\textsc{true}}}\xspace}
\def\tfalse{\ensuremath{\text{\normalfont\textsc{false}}}\xspace}

\def\Done#1#2{\ensuremath{\mathsf{\ifcsname DoneName#1\endcsname\csname DoneName#1\endcsname\else Done#1\fi}_{#2}}}
\def\DoneDef#1#2{\expandafter\def\csname DoneName#1\endcsname{#2}}
\DoneDef{1}{leaderDone}
\DoneDef{2}{searchDone}

\Vardef{DoneSearch}
\Vardef{DoneLeader}
\Vardef{Phase}
\Vardef{ExactPhase}
\Vardef{Clock}
\Vardef{Ticks}
\Vardef{FirstTick}

\Vardef{Junta}
\Vardef{Leader}
\Vardef{Error}
\def\K#1{\ensuremath{k_{#1}}}
\def\L#1{\ensuremath{l_{#1}}}

\def\Counter#1{j_#1}
\Vardef{Flip}
\Vardef{Level}
\Vardef{Active}
\Vardef{FinalRound}
\Vardef{PreviousLoad}

\def\protocol#1{{\normalfont\textsc{#1}}\xspace}
\def\Search Protocol{\protocol{Search} Protocol\xspace}

\def\matrixbracex#1#2{\underbrace{#1}_{\text{\makebox[0pt][c]{\footnotesize #2}}}}

\usepackage{dcmtitle}

\title{On Counting the Population Size}
\author[1]{Petra Berenbrink}
\author[1]{Dominik Kaaser}
\author[2]{Tomasz Radzik}
\affiliation{Universität Hamburg\\Hamburg, Germany\\\textit{\{petra.berenbrink, dominik.kaaser\}@uni-hamburg.de}}
\affiliation{King's College London\\London, UK\\\textit{tomasz.radzik@kcl.ac.uk}}
\begin{document}
\maketitle

\begin{abstract}
We consider the problem of counting the population size in the \emph{population
model}. In this model, we are given a distributed system of $n$ identical
agents which interact in pairs with the goal to solve a common task. In each
time step, the two interacting agents are selected uniformly at random. In this
paper, we consider so-called \emph{uniform} protocols, where the actions of two
agents upon an interaction may not depend on the population size $n$. We
present two population protocols to count the size of the population: protocol
\protocol{Approximate}, which computes with high probability either $\floor*{\log n}$
or $\ceil*{\log n}$, and protocol \protocol{CountExact}, which computes the exact population
size in optimal $\BigO{n\log{n}}$ interactions, using $\BigOTilde{n}$ states.
Both protocols can also be converted to \emph{stable} protocols that give a
correct result with probability $1$ by using an additional multiplicative
factor of $\BigO{\log{n}}$ states.

\end{abstract}

\section{Introduction} \label{sec:introduction}

In this paper we consider the problem of counting the population size in the
probabilistic \emph{population model}. The model was introduced in
\cite{DBLP:journals/dc/AngluinADFP06} to model distributed systems of
resource-limited mobile agents, which interact with each other in order to
solve a common task. The computation of a probabilistic \emph{population
protocol} can be viewed as a sequence of pairwise interactions of randomly
chosen agents. In each interaction, the two participating agents observe each
others' states and update their own state according to a transition function
common to all agents. 

In this setting, we are interested in so-called \emph{uniform} population
protocols, where the transition function does not depend on the size of the
population, so it can be applied to any population size. In the original
definition of population protocols in \cite{DBLP:journals/dc/AngluinADFP06},
each agent is a copy of the same finite state machine. Such a protocol is by
definition uniform, as its state space has constant size (independent of the
size of the population). However, more recent results have shown that many
problems can be solved much faster if agents are equipped with a number of
states that grows with the population size. For example, a simple protocol with
a constant number of states solves the majority problem in expected
$\BigO{n^2}$ interactions \cite{DBLP:journals/siamco/DraiefV12,
DBLP:journals/dc/MertziosNRS17}. On the other hand, a number of protocols have
been proposed which solve the majority problem in $\BigO{n\polylog n}$
interactions using $\BigO{\polylog n}$ states
\cite{DBLP:conf/soda/AlistarhAEGR17, DBLP:conf/soda/AlistarhAG18,
DBLP:conf/podc/BilkeCER17, DBLP:conf/wdag/BerenbrinkEFKKR18}. These protocols
are all non-uniform since their transition functions refer to values which have
to be $\Theta(\log n)$. We note that $\BigOmega{n\log n}$ interactions are
required to reach a positive constant probability that each agent participates
in at least one interaction, so $\BigOmega{n\log n}$ is a lower bound on the
number of interactions of a protocol which solves any nontrivial problem.

Turning to the problem of counting the population size, either exactly or
approximately, we naturally have to consider population protocol models with
growing memory, where lower bounds on the required number of states depend on
the desired counting accuracy. The exact counting of the population size
requires $\BigOmega{n}$ states, while estimating the population size up to a
constant factor requires $\BigOmega{\log n}$ states. 

There is a simple and uniform protocol for exact population counting, which
completes in expected $\Theta(n^2)$ interactions  and uses $\Theta(n^2)$
states: the agents start with one token each and keep combining the tokens into
bags, propagating at the same time the maximum size of a bag and using that
maximum as their current output. The transition function of this protocol does
not refer in any way to (any estimate of) the size of the population.

In this paper we are interested in uniform population counting protocols which
run in $\BigO{n\polylog n}$ interactions and use a small number of states
(relatively to the lower bounds). We present new protocols for exact and
approximate population counting, which use a significantly smaller number of
states than the previously best protocols shown in
\cite{DBLP:journals/corr/abs-1808-08913,DBLP:conf/wdag/DotyEMST18}.
Our protocols have also further desirable properties. Our exact protocol is the
first protocol for this problem which completes the computation within
asymptotically optimal $\BigO{n\log{n}}$ interactions. A variant of our
approximate protocol is the first $\BigO{\polylog n}$-state, $\BigO{n\polylog
n}$-interaction protocol \emph{always} converging to a value which is within a
constant factor from the exact size of the population. To give formal
statements of the previous results and our new results, we have to introduce
first some details of the model and the way the efficiency of population
protocols is measures.

\subsection{Computation Model} \label{sec:related-work}

The computation model is a population of $n$ agents which are capable of performing local
computations. 
A population protocol is specified by a state space $Q$, an output domain $O$,
a transition function $\delta: Q\times Q \rightarrow Q\times Q$,
and an output function $\omega: Q \rightarrow O$.
Each agent has a state $q \in Q$, which is updated during interactions.
The current output of an agent in state $q$ is $\omega(q)$.
The current \emph{configuration} of the system 
is the vector from $Q^n$ with the current states of the agents.
The computation of a population protocol is a sequence of pairwise interactions of agents.
In every time step, 
a \emph{probabilistic scheduler} selects independently and uniformly at
random a pair of agents for interaction, with the first 
agent called the \emph{initiator} and the second the \emph{responder}.
During the interaction the
agents update their states by applying the transition function $\delta$. 
Such an update is
denoted by $(x, y) \rightarrow (x', y')$, where $(x,y)$ refers to the states
of the agents
before the interaction and $(x', y')= \delta(x,y)$ to the states after the interaction. 

A given problem which we want to solve by population protocols
specifies the set of initial (input) configurations, the output domain $O$, and
the desired (output) configurations for given input configurations. 
For the exact population counting problem, all agents are initially in the same state $q_0$ and 
the output domain is the set of positive integers. 
A desired configuration is when all agents output correctly the (exact) size of the population.
In the \emph{leader election} problem, 
which arises as a sub-problem in many population protocols, including in the protocols
for the population counting problem,
all agents start with the same initial state 
and the output domain is $O=\set{\text{\emph{leader}}, \text{\emph{follower}}}$.
A desired configuration is when exactly one agent outputs that it is the leader.

The following two definitions are commonly used 
to capture important aspects of the time efficiency of population protocols.
 The \emph{convergence time $T_C$} of an execution of a protocol is the
number of interactions until the system enters a desired configuration and never leaves
the set of desired configurations again.
The \emph{stabilization time $T_S$} of an execution of a protocol is the
number of interactions until the system enters a desired \emph{stable} configuration,
meaning that starting from this configuration no
sequence of pairwise interactions can take the system outside of the set of desired configurations.
We always have $T_C \leq T_S$ 
for any execution of a protocol,
but the stabilization time may be strictly greater than the convergence time.
A population protocol is \emph{always correct} (or stable), if for each initial configuration
the computation reaches a 
correct stable configuration with probability $1$, or \emph{w.h.p.\ correct}, if there 
is a small (of the order of $n^{-\BigOmega{1}}$), but positive, probability that the system settles with an incorrect output or does not 
settle at all.
We say that a protocol converges (resp.\ stabilizes) in $T(n)$ time w.h.p.\ (resp.\ in expectation),
if $T_C$ (resp.\ $T_S$) is at most $T(n)$ \whp/ (resp.\ in expectation).

The second measure of efficiency of population protocols is the \emph{required number of states}. This is
a straightforward notion for simple protocols with constant number of states.
For more complex protocols, transition functions 
are usually described by pseudo-codes,
which refer to \emph{variables}.
The state space of a protocol which is defined in this way 
is the Cartesian product of the ranges of the variables.
If we consider all possible executions of a uniform protocol, then
the ranges of some variables may be very large in terms of $n$, potentially infinite.
We are, however, interested in bounds on the ranges of the variables (and thus bounds 
on the whole state space) that hold \whp/.

A population protocol is called \emph{uniform} if the same transition function
is used for all population sizes \cite{DBLP:journals/dc/AngluinAER07}. 
Uniformity of a protocol is thus a desired property since such protocols can be
applied without knowing the size of the
population in advance. 
The original model from 
\cite{DBLP:conf/podc/AngluinADFP04,DBLP:journals/dc/AngluinADFP06}
is uniform since the
number of states of an agent is constant.

\subsection{Related Work}
It is shown in \cite{DBLP:conf/podc/AngluinADFP04,DBLP:journals/dc/AngluinADFP06}
that with a constant number of states
all semilinear predicates (which include, e.g., the majority predicate),
can be computed in expected $\BigO{n^2 \log n}$ stabilization time.
Recently Kosowski and Uznanski \cite{DBLP:conf/podc/KosowskiU18}
have shown constant-state protocols 
for computing the semilinear predicates and for electing a leader, 
which \whp/ have $\BigO{n\polylog n}$ convergence time.
Achieving fast $\BigO{n\polylog n}$ stabilization time for the majority and leader election problems 
requires a growing state space,
as $\BigO{n^{2-\epsilon}}$ stabilization time is not possible with a constant number of states
\cite{DBLP:journals/dc/DotyS18,DBLP:conf/soda/AlistarhAEGR17}.

There are a number of
protocols for the majority problem and the leader election problem
with $\BigO{n\polylog n}$ stabilization time and $\BigO{\polylog n}$ states
\cite{DBLP:conf/podc/AlistarhGV15,DBLP:conf/soda/AlistarhAEGR17,DBLP:conf/podc/BilkeCER17,DBLP:conf/soda/AlistarhAG18,DBLP:conf/wdag/BerenbrinkEFKKR18},
but they all are non-uniform.
The transition functions of these protocols depend
on $n$ as they use values $\BigOmega{\log n}$ or $\BigOmega{\log\log n}$.
For example, each agent may be counting its interactions (which are uniform updates, not dependent on $n$)
and progress to the next phase of the computation when its counter reaches $\log n$
(a non-uniform update).

While the computational power of (uniform) constant-space population protocols 
is well understood by now, less is known about the power of uniform protocols which 
allow the state space to grow with the size of population. 
Doty et al.\ \cite{DBLP:conf/wdag/DotyEMST18} 
show a protocol for the exact population counting problem
which has $\BigO{n\log n \log\log n}$ stabilization time and uses $\BigO{n^{60}}$ states,
\whp/ and in expectation.
They also formalize the notion of uniform population protocols, modeling the agents 
as copies of the same Turing machine, but describe their protocols using pseudo-codes 
and take the size of the state space as the product of the ranges of the variables.

Doty and Eftekhari~\cite{DBLP:journals/corr/abs-1808-08913} consider the problem 
of estimating the size of population within a constant factor, which they view as 
the problem of computing $\log n \pm \BigO{1}$.
They show a protocol which \whp/ has $\BigO{n \log^2 n}$ convergence time 
and uses $\BigO{\log^7 n \log\log n}$ states.
Their protocol is not always correct and they left open the question 
to design a protocol which uses expected $\BigO{n \polylog n}$ interactions and $\BigO{\polylog n}$ states and computes 
$\log n \pm \BigO{1}$ with probability $1$.
An earlier work by Alistarh et al.\ \cite{DBLP:conf/soda/AlistarhAEGR17} includes 
a uniform protocol which uses $\BigO{n \log n}$ interactions and $\BigO{\log{n}}$ states in expectation to compute 
an integer which \whp/ is between $c_1 \log n$ and $c_2 \log n$, for some constants 
$0 < c_1 < 1 < c_2$. This can be viewed as approximating the population size within a polynomial factor.

\subsection{Our Results} \label{sec:model}
In this paper we present and analyze two uniform protocols, protocol \protocol{Approximate}
for approximating the population size within constant factors, and protocol \protocol{CountExact}
for computing the exact number of agents in the population. Unless stated otherwise we assume that all
agents have to output the population size or its approximation.
The theorems below summarize the performance of our protocols.

\begin{theorem} \label{thm:approx}
Protocol \protocol{Approximate} is uniform and outputs \whp/
either $\smash{\floor*{\log{n}}}$ or $\smash{\ceil*{\log{n}}}$.
\begin{enumerate}
\item \label{thm:approx-1}
It converges in at most $\operatorname{O}(n\log^2{n})$ interactions using
$\operatorname{O}(\log{n}\cdot\log\log{n})$ states, \whp/.

\item \label{thm:approx-2}
A variant of the protocol stabilizes in $\operatorname{O}(n\log^2{n})$ 
interactions using at most $\operatorname{O}(\log^2{n} \cdot \log\log{n})$ states, \whp/.
\item \label{thm:approx-3}
Stabilization can also be achieved \whp/ in $\operatorname{O}(n\log^2{n})$ 
interactions using $\operatorname{O}(\log {n} \cdot \log\log{n})$ states,
if not all but only $n - \log{n}$ agents need to output the result.
\end{enumerate}
\end{theorem}

\begin{theorem}[restate=restatethmcounting,label=thm:counting]
The protocol \protocol{CountExact} is uniform and outputs the exact population
size $n$. It stabilizes in $\BigO{n\log{n}}$ interactions and uses
$\BigOTilde{n}$ states, \whp/.\footnote{%
We define $\BigOTilde{f(n)}=\BigO{f(n)\cdot\log^{\BigO{1}}(f(n))}$.%
}
\end{theorem}

Our protocols improve considerably the time and space bounds of the 
previous work \cite{DBLP:conf/wdag/DotyEMST18} and 
\cite{DBLP:journals/corr/abs-1808-08913}.
Moreover, our approximate protocol is stable, answering the open question posed 
in~\cite{DBLP:journals/corr/abs-1808-08913}, and it calculates a tighter approximation  
converging to $\floor*{\log{n}}$ or $\ceil*{\log{n}}$ instead of $\log n \pm 4.7$ shown in \cite{DBLP:journals/corr/abs-1808-08913}.
The stabilization time of our exact counting protocol is asymptotically optimal  
and its number of states is only a $\polylog n$ factor away from the lower bound $n$.
We note that in our protocol for exact counting the output value of an agent is a function of 
the state of this agent, and $n$ is not kept explicitly in the agent's state.
This is consistent with the approach in \cite{DBLP:conf/wdag/DotyEMST18}.

Our algorithms are based on leader election followed by load balancing phases. 
A load balancing phase starts with the total load of $M$ tokens in the system
and the goal is to relate $M$ to $n$ (the size of the population). 
If this load balancing phase completes with some agents having zero load, then
\whp/ $M \leq cn$, for a fixed constant $c > 1$.
On the other hand, if all agents end up with positive load, then we must have $M \geq n$.
To achieve good time and space bounds, we have to carefully control the load balancing phases
and integrate them with the process of leader election.
To achieve stability both of our protocols use an error detection routine
and the leader initializes a
process testing if the calculated answer is correct or not.

In \cref{sec:auxiliary-protocols} we discuss the auxiliary protocols which we use 
in our main protocols.
We outline our approximate and exact protocols in \cref{sec:approx,sec:counting}, 
respectively.

\section{Auxiliary Protocols} \label{sec:auxiliary-protocols}
In this section we define auxiliary protocols that we will apply in our
protocols: one-way epidemics, the junta process, leader election, and phase
clocks.

\paragraph{One-Way Epidemics}
The goal of one-way epidemics is to spread an information to all members of the population.
The state space of the protocol is $\set{0, x}$ for some $x > 0$. Initially at least one
agent has the value $x$, which is then spread to all other agents. The
transitions are formally defined as
$\delta(u, v) = (\max\set{u, v}, v)$.
We will refer to one-way epidemics also as \emph{broadcast}.
A natural extension is \emph{maximum broadcast}, where agents do not only spread one
possible value $x$, but instead each agent starts with its own value from the integer
interval $[0, x]$. For maximum broadcast, an agent always adopts the maximum,
which is covered by the transition rule above as well.
The result on one-way epidemics carries over immediately to
maximum broadcast. The following result is well-known, see for example \cite{DBLP:journals/dc/AngluinAE08a}.
\begin{lemma} \label{lem:broadcast}
Let $\T_{BC}$ be the number of interactions required to complete (maximum)
broadcast. \Whp/, $\T_{BC} = \BigO{n \log{n}}$.
\end{lemma}

\paragraph{Junta Process}\label{x-sec:junta}
The goal of the \emph{junta process} \cite{GS18, DBLP:journals/corr/abs-1805-04586} is to mark
$\BigTheta{n^\epsilon}$ agents -- the \emph{junta}. The state of each agent $v$
in this process is formed by a triplet $(\Level v, \Active v, \Junta v) \in \mathbb{N}_0
\times \set{0,1}^{2}$ initially set to $(0, 1, 1)$. The idea of the protocol is as follows. If an active agent $v$ interacts with an active agent on the same level it increases its level, otherwise it
sets $\Active v$ to $0$. Whenever $v$ interacts with an agent on a higher level it sets $\Junta v$ to $0$.  Inactive agents adopt the level of their communication partner if that is higher. 
The protocol stabilizes when all agents are inactive.
The \emph{junta} is formed by all agents $v$ that reached the maximum level
and have their $\Junta v$ bit set to $1$. The following lemma 
is shown in \cite{DBLP:journals/corr/abs-1805-04586}, here it is adapted to our setting.

\begin{lemma} (\cite{DBLP:journals/corr/abs-1805-04586}) \label{lem:junta}
Let $\Level* = \max\set{\Level v}$ be the maximal level reached by the
junta process. 
All agents become inactive within $\BigO{n\log{n}}$ interactions,
$\log\log{n} - 4 \leq \Level* \leq \log\log{n} + 8$, and
the number of agents on the maximal level is $\BigO{\sqrt{n}\cdot\log{n}}$, \whp/.
\end{lemma}

\paragraph{Phase Clocks} \label{x-sec:phase-clocks}
The existence of a non-empty junta of size in $[1, n^{\epsilon}]$ allows the agents to synchronize
themselves via so-called \emph{phase clocks} \cite{DBLP:journals/dc/AngluinAE08a,GS18}.
The phase clocks allow all agents to divide the time in phases of $\BigTheta{n\log{n}}$
interactions each. All agents $v$ have a
state $\Clock v$ in $\set{0, \dots, m-1}$ where $m$ is a constant. Intuitively, these clock states can be seen as the hours on a clock face. 
The basic idea in every interaction is that the
agents always adopt the larger clock state w.r.t.\ the circular order modulo $m$.
Additionally, in order to keep the clock running, the members of the junta proceed one
additional step when they interact with another agent on the same clock state.
An agent $v$ enters a new phase in interaction $t$
if $\Clock v$
crosses the boundary between $m-1$ and $0$.
In that case we say the phase clock \emph{ticks}.

For easy access to the phase clocks, we equip each agent $v$ with a variable
$\Phase v$ of constant size that counts the current phase of an
agent modulo some constant. 
Additionally, each agent
$v$ has a flag\footnote{Note that for all flags we use $\tfalse$ (resp.\ $\ttrue$)
and $0$ (resp.\ $1$) interchangeably.} $\FirstTick v$.
This flag is set to $1$ whenever the $\Phase v$
counter is incremented, and it is set to $0$ otherwise.

\def\Din#1{D^{\text{\normalfont start}}_{#1}}
\def\Dout#1{D^{\text{\normalfont end}}_{#1}}
Let $D_{i} = [\Din{i}, \Dout{i}]$ be the interval of interactions in
Phase $i$ such that the last agent enters Phase $i$ in interaction $\Din{i}$
and the first agent leaves Phase $i$ in interaction $\Dout{i}+1$.
\begin{lemma}[\cite{GS18}] \label{lem:phase-clocks}
For any constant $c\geq 0$ we can construct a phase clock using $m = m(c) = \BigO{1}$
states such that \whp/ for all phases $1\leq i \leq \poly(n)$ 
we have $c n\log{n} \leq \Dout{i} - \Din{i} \leq c n\log{n} + \BigTheta{n\log{n}}$.
\end{lemma}

\paragraph{Leader Election}\label{x-sec:leader-election}
In \cite{GS18} the authors present a stable and uniform protocol to perform
leader election called \protocol{leader\_elect}. It runs in $\BigO{\log{n}}$
phases of $\BigO{n\log{n}}$ interactions each. 

The process starts with junta election in order to start two nested phase clocks. 
Agents perform an
interaction of the outer phase clock once per phase of the 
inner phase clock.
In the beginning, every agent runs the protocols for the phase
clocks, junta election, and leader election in parallel. Whenever an agent
encounters another agent on a higher (junta) level, it resets the clocks and the
leader election protocol. In that way, all agents eventually run the phase
clocks and the leader election process based on the junta on the highest level.
The actual leader election process is quite simple: the set of leaders is halved from
phase to phase, and the time is measured by the inner phase clock.

For stable leader election the authors of \cite{GS18} combine their protocol with a slow 
protocol which is always correct in the following way:  
the outer phase clock runs only if at least one leader remains
in the system.
If at some point it counts to $m$, \whp/ a total of $\BigTheta{n\log^2{n}}$ 
interactions have occurred. At this time, all agents override their current leader 
state from the slow protocol with that of the fast protocol.
When the outer phase clock ticks (i.e., it reaches $m$), at least one leader
exists. We equip each agent $v$ with an additional flag
$\Done1 v$, initially set to $\tfalse$. It is set to \ttrue as soon as the
outer phase clock ticks.

\begin{lemma}[\cite{GS18}] \label{lem:leader-election}
The uniform protocol \protocol{leader\_elect} elects a unique leader. It stabilizes
in $\BigO{n\log^2{n}}$ interactions, using  $\BigO{\log\log{n}}$
many states, \whp/. Furthermore, after at most $\BigO{n \log^2{n}}$ interactions all
agents $v$ have $\Done1 v$ set to \ttrue \whp/, and at that time there is exactly one leader \whp/.
\end{lemma}

\paragraph{Fast Leader Election} \label{x-sec:fast-leader-election}
In \cite{DBLP:journals/corr/abs-1805-04586} the authors describe a stable
leader election protocol that admits a trade-off between the running time and
the number of states. When using $\BigOTilde{n}$ states, their protocol
stabilizes \whp/ in $\BigO{n\log{n}}$ interactions. In the following, we call
this protocol \protocol{FastLeaderElection}. The main idea of the protocol is to use 
$\BigO{\log n}$ random bits to reduce the number of active leaders much faster than in the original protocol from \cite{GS18}.

\begin{lemma}[restate=restateFastLeaderElection,label=lem:fast-leader-election,name={\cite{DBLP:journals/corr/abs-1805-04586}}]
The uniform protocol \protocol{FastLeaderElection} elects a unique leader. It
stabilizes in $\BigO{n\log{n}}$ interactions, using $\BigOTilde{n}$
many states, \whp/. Furthermore, after at most $\BigO{n \log^2{n}}$ interactions all
agents have $\Done1 u$ set to \ttrue \whp/.
\end{lemma}

\section{Approximate Counting } \label{sec:approx}

In this section we show \cref{thm:approx}. In \cref{sec:search-protocol} we
assume that a unique leader exists and we present a protocol which calculates
$\floor*{\log n}$ or $\ceil*{\log n}$ using a load balancing algorithm. In
\cref{sec:analysis-search-protocol} we analyze the protocol. In
\cref{sec:putting-things-together} we show how to combine the protocol from
\cref{sec:search-protocol} with a leader election protocol. This shows the
first part of \cref{thm:approx}. Finally, in \cref{sec:stable-short} we show
how to build a stable protocol by showing the correctness of our error
detection mechanism. This shows the second and the third part of \cref{thm:approx}.

\subsection{Approximating $n$ with a Leader} \label{sec:search-protocol}
In this section we assume that a unique leader is given and that all agents are
synchronized via the phase clocks. The main idea of our algorithm is to inject
an increasing amount of tokens into the system and to use a load balancing
routine to estimate the number of agents in the system. The process is finished
as soon as roughly $n$ tokens are injected into the system. To save on the
number of states, the agents do not store the exact number of tokens, they hold
but the logarithm of that number.

Every agent $v$ stores two variables $\left(\K v, \Done2 v\right) \in \set{-1,
0, 2, \dots} \times \set{0,1}$, initially set to $(-1, 0)$. The variable $\K v$
stores the logarithm of the \emph{load} of agent $v$, where the special value
$-1$ indicates that the agent is empty. The variable $\Done2 v$ indicates that
the search is finished. The leader $u$ now orchestrates a linear search over
$\K u \in \set{0, 1, \dots }$ and finds $\K u^*$ for which $\log{n} - 1 < \K
u^* < \log{n} +1$. The protocol runs in rounds consisting of multiple phases
each. At the beginning of round $r$ the leader injects $2^{r}$ many tokens
into the system. The load is then balanced using a \emph{powers-of-two load
balancing} process which restricts the load of any agent to a power of two. The
search stops once an agent $v$ has load larger than $1$ ($\K v > 0$). At this
time the leader $u$ sets $\Done2 u$ to \ttrue and broadcasts the value to all
agents using one-way epidemics. 

\medskip

In the classical load balancing process (cf.\ \cite{BFKK19}), it is assumed
that $m$ indistinguishable tokens are distributed arbitrarily among $n$ agents.
At interaction $t$, the \emph{load vector} $\mathbf{L}(t)$ is defined as
$\mathbf{L}(t) = \left(\ell_1(t), \dots, \ell_n(t) \right) \in \N_{0}^{n}$,
where $\ell_i(t)$ is the number of tokens (load) present at agent $i$ at that
interaction. Assume agents $u$ and $v$ 
balance their loads in interaction $t\geq1$, then
$\left( \ell_u(t+1), \ell_v(t+1) \right) = \left( \floor*{(\ell_u(t) + \ell_{v}(t))/2}, 
\ceil*{(\ell_u(t) + \ell_v)/2} \right)$.
For the powers-of-two load balancing process we define the \emph{logarithmic load vector}
$\mathbf{K} = \left(k_1, \dots, k_n \right) \in \set{-1, 0, 1,
\dots}^n$ such that for any agent $v$ we have
$\ell_v(t) = 2^{k_v}$ if $k_v \geq 0$ and $\ell_v(t) = 0$ if $k_v = -1$.
Assume now that agent $u$ interacts with agent $v$. A balancing action is only permitted  
if either $u$ or $v$ is empty (i.e., $\K u=-1$ or $\K v=-1$) and both of them are not the leader. 
Let $k'_u$ and $k'_v$ be the resulting logarithmic load values.
Then
\begin{equation} \label{eq:powers-of-two-balancing}
\left( k'_u, k'_v \right) =
\begin{cases}
\left( k_u-1, k_u -1 \right) & \text{if } k_u > 0 \text{ and } k_v = -1 \\
\left( k_v-1, k_v -1 \right) & \text{if } k_u = -1 \text{ and } k_v > 0 \\
\left( k_u, k_v \right) & \text{otherwise.}
\end{cases}
\end{equation}

\begin{lstalgo}[h]{The \Search Protocol, centerpiece of protocol \protocol{Approximate}\label{alg:search}}
Interaction $(u, v)$ in the #\Search Protocol#:
	if $\Leader u = \ttrue$ and $\Done2 u = \tfalse$ then #\comment{Leader:}#

		if $\Phase u = 1$ and $\FirstTick u = \ttrue$ then #\comment{Phase 1: load infusion}#
			$\K v \gets \K u$#\label{ln:load-infusion}#

		if $\Phase u = 4$ and $\FirstTick u = \ttrue$ then #\comment{Phase 4: decision}#
			if $\K v \leq 0$ then#\label{ln:decision}#
				$\K u \gets \K u + 1$
			else
				$\Done2u \gets \ttrue$#\label{ln:conclude-search}\medskip#

	if $\Leader u = \tfalse$ and $\Leader v = \tfalse$ then #\comment{Followers:}#

		if $\Phase u = 0$ then #\comment{Phase 0: initialize}#
			$\K u \gets -1$ #\label{ln:reset-to--1}#

		if $\Phase u = 2$ then #\comment{Phase 2: load balancing}#
		
			if $\min\set{\K u, \K v} = -1$ and $\max\set{\K u, \K v}>0$ then 
				$\K u, \K v \gets \max\set{\K u, \K v} - 1$#\label{ln:load-balancing}#

		if $\Phase u = 3$ then #\comment{Phase 3: one-way epidemics}#
			$\K u, \K v \gets \max\set{\K u, \K v}$#\label{ln:broadcast}#
\end{lstalgo}
The \Search Protocol is defined in \cref{alg:search}. The protocol runs in
rounds of 5 phases. Every agent $v$ uses the variable $\Phase{v}$ to count the number of
phases modulo $5$ and the flag $\FirstTick v$ which is set to \ttrue when $v$
initiates the first interaction of a phase (see \cref{sec:auxiliary-protocols}). In
Phase $1$ and Phase $4$ the leader is active, in the remaining phases the
non-leaders are active. 

\begin{itemize}
\item Phase $0$ and Phase $1$ are used for the initialization. 
Every agent $v$ which is not the leader resets $\K v$ to $-1$ (\cref{ln:reset-to--1}).
During the first interaction $(u,v)$ in Phase $1$ the leader $u$ transfers $2^{\K u}$ tokens to $v$
(\cref{ln:load-infusion}).
\item In Phase $2$, the non-leader agents perform powers-of-two load balancing, as described above.
(\cref{ln:load-balancing}).
Phase $3$ consists of one-way epidemics where the agents communicate their highest load value (\cref{ln:broadcast}).
\item In Phase $4$ (\cref{ln:decision}), the leader $u$ decides if the search is finished or not. 
If the maximum logarithmic load is less than $1$ it concludes that the injected load was smaller than or equal to $\log{n} - 1$. 
The leader $u$ therefore proceeds to the next round and injects twice as many tokens.
If the observed maximum logarithmic
load is larger than $1$, the leader concludes the protocol by setting
$\Done2 u$ to $\ttrue$.
\end{itemize}
The number of agents is now estimated as~$2^{\K u}$.
In the following, the state of an agent $v$ in this algorithm is called $s(v)$.
The state space $S$ is the Cartesian product of the domains of the individual
variables ($\Phase v$, $\Phase v$, $\FirstTick v$, $\Leader v$, $\K v$, $\Done2
v$). For an overview over the states of node $v$ in the \Search Protocol, see 
\cref{fig:search-state-space}.

\begin{figure}[ht]
\framebox[\columnwidth]{
\scalebox{0.8}{
\begin{minipage}{\columnwidth}
\centering
\begin{equation*}
\setlength\arraycolsep{0pt}
\begin{array}{rccccc}
s(v) ={}& \Big( 
\matrixbracex{\Phase v,~ \FirstTick v,}{Phase Clocks}
&
& \matrixbracex{\Leader v,}{Leader Election}
&
& \matrixbracex{\K v,~ \Done2 v }{\Search Protocol}
\Big) \\
S ={} &
\set{0, \dots, 4 }\times \set{0,1}  &
{}\times{} &
\set{0,1} &
{}\times{} &
{\set{-1, 0, 1, \dots} \times \set{0,1}}
\\
\end{array}
\end{equation*}
\end{minipage}
}
}

\caption{State space of the \Search Protocol (\cref{alg:search})}\label{fig:search-state-space}
\end{figure}

\subsection{Analysis of the \Search Protocol} \label{sec:analysis-search-protocol}

Based on the definition in \cref{eq:powers-of-two-balancing} we first show that the
powers-of-two load balancing process on $n$ agents balances the load such that the maximum load is at most $1$ (i.e., $\K v \leq 0$ for all agents~$v$) as long as at most $3/4 n$ tokens are injected by the leader.
\begin{lemma}[restate=restateLoadBalancing,label=lem:powers-of-two-load-balancing]
Assume there exists an agent $u$ with $k_u(0) = \kappa$ and $k_v(0) = -1$ for all $v \neq u$.
If $2^{\kappa} \leq 3/4\cdot n$, then \whp/ after $t=16 n \log n$ interactions $\max_v\set{k_v(t)} = 0$.
\end{lemma}
The proof of \cref{lem:powers-of-two-load-balancing} uses the same ideas and
shows a similar statement
as in Lemmas~2 and~3 from \cite{BFKK19} for the classical load balancing
process.

\newcounter{lstnumberAlgApproximate}
\begin{lstalgo}[t]{Protocol \protocol{Approximate}\label{alg:approx}}
Interaction $(u, v)$ of protocol #\protocol{Approximate}#:
	if $\Level v > \Level u$ then
		#re-initialize states of $u$ for \protocol{PhaseClocks}, \protocol{LeaderElection}, \protocol{SearchProtocol}#

	$\text{\protocol{JuntaProcess}}\,[(\Level u, \Junta u),\, (\Level v, \Junta v)]$ #\comment{update auxiliary protocols}#
	$\text{\protocol{PhaseClocks}}\,[(\Phase u, \FirstTick u),\, (\Phase v, \FirstTick v)]$
	

	if not $\Done1 u$ then #\comment{Stage 1: Leader Election \cite{GS18}}#
		$\text{\protocol{LeaderElection}}\,[(\Leader u, \Done1 u),(\Leader v, \Done1 v)]$

	if $\Done1 u$ and not $\Done2 u$ then #\comment{Stage 2: Search Stage}#
		$\text{\protocol{SearchProtocol}}\,[(\K u, \Done2 u), (\K v, \Done2 v)]$

	if $\Done1 u$ and $\Done2 u$ then #\comment{Stage 3: Broadcasting Stage}\setcounter{lstnumberAlgApproximate}{\value{lstnumber}}#
		$(\Done2 v, \K v) \gets (\ttrue, \K u)$
\end{lstalgo}
\begin{figure*}[ht]
\framebox[\textwidth]{
\scalebox{0.8}{
\begin{minipage}{\textwidth}
\centering
\begin{equation*}
\setlength\arraycolsep{0pt}
\begin{array}{lccccccc}
\text{state of node $v$:}\quad & \Big(
\matrixbracex{
\Level v,~
\Active v,~
\Junta v,}{Junta Process}
&
& \matrixbracex{
\Clock v,~
\Phase v,~
\FirstTick v,}{Phase Clocks}
&
& \matrixbracex{
\Leader v,~
\Done1 v,}{Leader Election}
&
& \matrixbracex{s(v)}{\Search Protocol}
\Big)
\\
\text{state space:}&
\mathbb{N}_0 \times \set{0,1}\times\set{0,1} &{}\times{}&
[m] \times \set{0,\dots,4}\times\set{0,1} &{}\times{}&
\set{0,1}\times\set{0,1} &{}\times{}&
 S\phantom{{}\Big)}
\end{array}
\end{equation*}
\end{minipage}
}
}
\caption{State space of protocol \protocol{Approximate} (\cref{alg:approx})}\label{fig:apx-state-space}
\end{figure*}

In the following, we assume that precisely one agent $u$ is the leader and all
agents synchronize themselves via the phase clocks. Based on this assumption we
can show the following lemma for the \Search Protocol.
\begin{lemma} \label{lem:search}
After at most $\BigO{\log{n}}$ rounds the leader $u$ sets $\Done2u$ to
\ttrue, \whp/.  At that time we have \whp/ that
$3/4\cdot n < 2^{\K u} \leq 2^{\ceil*{\log{n}}}$.
\end{lemma}
\begin{proof}
The \Search Protocol runs in multiple phases. In the following, we combine every
five consecutive phases to a round such that round $r\geq0$ consists of phases
$5r$ to $5r+4$. For a fixed round $r$ let $\K v(r)$ be the value of
variable $\K{}$ of agent $v$ at the beginning of round $r$.
Assuming that all agents are properly synchronized by the phase clocks (see \cref{lem:phase-clocks}),
we can \whp/ assume that
the number of interactions in any round is sufficient for spreading information with one-way 
epidemics (see \cref{lem:broadcast}) and for the powers-of-two load balancing (see \cref{lem:powers-of-two-load-balancing}).

Observe that the leader sets $\Done2u$ to \ttrue in \cref{ln:conclude-search}
of Phase 4 of the \Search Protocol. Let $v$ be an agent with maximal load at
the beginning of Phase 3. In Phase 3 the agents used one-way epidemics to
disseminate the value $k_v$ to all agents. Hence, $\Done2u$ is set to \ttrue in
the first round $r$ where an agent $v$ has $k_{v}(r)>0$ after the load
balancing, meaning agent $v$ has a load of at least two.

To show the lemma we consider 3 cases, depending on the value $\K u(r)$ which 
corresponds to the load injected by the leader at the beginning of round $r\geq0$.

\paragraph{Case $2^{\K u(r)} \leq 3/4\cdot n$:\?}
In Phase $1$ the leader injects $2^r$ many tokens. 
According to \cref{lem:powers-of-two-load-balancing}, \whp/ at the end of Phase
$2$ of Round $r$ no agent will have a load larger that one, meaning for all agents $v$
we have $k_v\leq 0$. The maximum value $k_v$ which is communicated via one-way epidemics is at most zero
and the leader does not set $\Done2u$ to $\ttrue$ in Phase 4 of Round $r$. Hence, it injects 
$2^{r+1}$ tokens at the beginning of the next round. 

\paragraph{Case $3/4 \cdot n < 2^{\K u(r)} < n$:\?}
Again, in Phase $1$ the leader injects $2^r$ many tokens. 
In this case either there exists an agent $v$ with a load larger than one ($k_{v}>0$) at the end of Phase 2,
or for all agents $v$ we have $k_{v}\leq 0$. In the latter case we are back to the previous case: the leader 
does not set $\Done2u$ to $\ttrue$ in Phase 4 of Round $r$. Hence, it injects 
$2^{r+1}$ tokens at the beginning of the next round. 
If there exists an agent $v$ with $k_{v}>0$ the leader sets $\Done2u$ to $\ttrue$. 

\paragraph{Case $n\leq 2^{\K u(r)}$:\?}
In Phase $1$ the leader injects $2^r\geq n$ many tokens. The agents balance at
least $n$ tokens on $n-1$ agents (since the leader does not participate in the
load balancing). Hence, at the end of Phase $2$ of Round $r$ there exists an
agent with a load of at least two, meaning $k_v\geq 1$. The leader sets
$\Done2u$ to $\ttrue$ in Phase 4 of Round $r$.

\medskip

From the above cases it follows that after at most $\ceil*{\log{n}}$ rounds 
$\Done2u$ is set to $\ttrue$ and then $3/4\cdot n < 2^{\K u} \leq 2^{\ceil*{\log{n}}}$.
\end{proof}

\subsection{Combining the \Search Protocol with Leader Election} \label{sec:putting-things-together}

In this section we combine the \Search Protocol with the
\protocol{LeaderElection} protocol from \cite{GS18}. The combined protocol
works as follows. All agents run the \protocol{JuntaProcess} and the
\protocol{PhaseClocks} protocols in parallel to the \protocol{LeaderElection}
protocol and, later, the \Search Protocol. In the junta process, agents cannot
decide whether they have already reached the maximal junta level. Therefore,
some agents already perform some interactions of the \protocol{LeaderElection}
protocol or the \Search Protocol at a lower junta level. These agents might be
badly synchronized, since the phase clock ticks correctly \whp/ only on the
maximal junta level. We therefore define the following procedure that has
already been used in \cite{GS18}: Whenever an agent $u$ interacts with an agent
$v$ in a higher level ($\Level u < \Level v$) all variables for
\protocol{PhaseClocks}, \protocol{LeaderElection} and the \Search Protocol are
re-initialized. In that way, all agents start the \protocol{LeaderElection}
protocol and the \Search Protocol at the maximal junta level from a clean
state.

The protocol uses two flags $\Done1v$ and $\Done2v$ for each agent $v$, which are initially set to \tfalse. The
first flag, $\Done1v$, is set once an agent has concluded the leader election
protocol (see \cref{sec:auxiliary-protocols}). The second flag, $\Done2v$, is set by the \Search Protocol (see \cref{sec:search-protocol}). The two flags $\Done1v$ and $\Done2v$ allow agent $v$ to distinguish between three
stages of the execution of the protocol, the \emph{Leader Election Stage} (1),
the \emph{Search Stage} (2), and the \emph{Broadcasting Stage} (3).

In the Leader Election Stage all agents use the protocol from \cite{GS18} to
elect a leader. Recall that the flag $\Done1 u$ is set to \ttrue once the
external phase clock ticks (see \cref{sec:auxiliary-protocols}). Note that at this
time there exists exactly one leader \whp/ \cite{GS18}.
In the Search Stage, we use the \Search Protocol defined in
\cref{sec:search-protocol}. In the Broadcasting Stage the leader uses one-way
epidemics to inform all other agents of its value $\K u$.

\medskip

We now give the proof of the first part of \cref{thm:approx}, where we show that
protocol \protocol{Approximate} computes either $\floor*{\log{n}}$ or
$\ceil*{\log{n}}$ \whp/. 

\begin{proof}[Proof of \cref{thm:approx-1} of \cref{thm:approx}]
From \cref{lem:junta} it follows that after $\BigO{n\log{n}}$ interactions the 
junta is elected and all agents are inactive. From that point on no agent in the protocol is ever
re-initialized again. Let $u$ be the single leader that
according to \cref{lem:leader-election} \whp/ 
concludes the leader election protocol. According to \cref{lem:search}, the
leader $u$ sets $\Done2 u$ to \ttrue after at most $\ceil*{\log{n}}$ phases of
$\BigO{n\log{n}}$ interactions each, and at that time $\K u$ is either
$\floor*{\log{n}}$ or $\ceil*{\log{n}}$ \whp/. According to
\cref{lem:broadcast}, within $\BigO{n \log{n}}$ further interactions all agents
know the value $\K u$ from the leader $u$. Together, this gives a total number
of $\BigO{n \log^2{n}}$ interactions, after which all agents know the value $\K
u$ from the leader and thus output either $\floor*{\log{n}}$ or
$\ceil*{\log{n}}$ \whp/.

Regarding the required number of states we observe that $\Level v$ from the junta process
and $\K v$ from the \Search Protocol
are the only variables of an agent $v$ that are not of constant
size and thus may grow with the population size $n$. For an overview over
the states of agent $v$, see \cref{fig:apx-state-space}. %
Note that $\Level v$
is bounded \whp/ by $\BigO{\log\log{n}}$ \cref{lem:junta}, and $\K v$
is bounded according to \cref{lem:search} \whp/ by $\BigO{\log{n}}$. Together this
gives that \whp/ the state space has size $\BigO{\log\log{n} \cdot \log{n}}$.
This concludes the proof of \cref{thm:approx-1} of \cref{thm:approx}.
\end{proof}

\subsection{The Stable Protocol} \label{sec:stable-short}
In this section we sketch our stable protocol and proofs of \cref{thm:approx-2,thm:approx-3} of \cref{thm:approx}. The details can be found in \fullversion. Our
stable protocol is a hybrid protocol which combines the protocol \protocol{Approximate} with
a \emph{slow} protocol which always finds the correct solution. 

The slow protocol works as follows. Every agent starts with exactly one token.
Whenever two agents interact and both have the same amount of tokens, one of
them hands its tokens over to the other. The agents do not store the exact
number of tokens they hold but the logarithm of that number. Note that, due to
the definition of the process, the load of every agent is a power of two. After
$\BigO{n^2\log{n}}$ many interactions one agent will store $\floor*{\log n}$
many tokens. For each $0\leq i<\floor*{\log{n}}$ there might be one agent
storing $2^{i}$ tokens (depending on the value of $n$). The number of states
required by this process is $\BigO{\log{n}}$ if it is sufficient for all but
$\log{n}$ agents to know the approximation for $n$. Otherwise, the value
$\floor*{\log n}$ has to be sent to every agent via one-way epidemics. In this
case there are up to $\log n$ many agents that need $\BigO{\log^2 n}$ many
states (exactly those agents which store a value $2^i$).

The hybrid protocol now works as follows. It first runs \protocol{Approximate},
then the leader $u$ checks if the approximated value $\K u$ is correct by
injecting $2^k$ tokens and by balancing them. If the value of $\K u$ is not
correct it switches the output over to the slow protocol. The hybrid protocol
stabilizes \whp/ within $\BigO{n \log^{2} n}$ many interactions using $\BigO{ \log\log n
\cdot \log^2 n }$ many states if all agents have to know the approximation of
$n$, otherwise $\BigO{ \log\log n \cdot \log n }$ states are sufficient.

\section{Counting the Exact Population Size} \label{sec:counting}

\def\ell{l}
\DoneDef{2}{ApxDone}

In this section we consider the problem of counting the exact population size
and present our protocol \protocol{CountExact}. The main idea is as follows.
The protocol first selects a junta using the modified junta process from
\cite{DBLP:journals/corr/abs-1805-04586} (see \cref{sec:auxiliary-protocols}). It then
creates phase clocks and selects a leader using the
\protocol{FastLeaderElection} protocol (see \cref{sec:auxiliary-protocols}).
Again, every agent $v$ has two flags $\Leader v$ and $\Done 1 v$ which indicate
whether $v$ is a leader and $v$ has completed the leader election protocol,
respectively. Recall that the junta uses a variable $\Level v$ that stores the
maximum level which is reached during the junta election. In
\cite{DBLP:journals/corr/abs-1805-04586} it has been shown that in the modified
junta process the level reaches \whp/ a value of $\log\log n \pm c$ for some
constant $c$. Hence, we can use $2^{2^{\Level v}}$ as a first approximation for
$n$. This approximation will be refined in two stages, called
\protocol{ApproximationStage} and \protocol{RefinementStage}, which are
described in \cref{sec:apx-stage} and \cref{sec:refinement-stage},
respectively. The protocol \protocol{CountExact} can be found in
\cref{alg:exact}. The output of every agent in the protocol
\protocol{CountExact} is defined by the \protocol{RefinementStage}.

\begin{lstalgo}[p]{Protocol \protocol{CountExact}\label{alg:exact}}
Interaction $(u, v)$ of protocol #\protocol{CountExact}#:
	if $\Level v > \Level u$ then
		#re-initialize states of $u$ for \protocol{PhaseClocks}, \protocol{FastLeaderElection}, \protocol{ApproximationStage}, \protocol{RefinementStage}#

	$\text{\protocol{JuntaProcess}}\,[\Level u,\, \Level v]$ #\comment{update auxiliary protocols}#
	$\text{\protocol{PhaseClocks}}\,[(\Phase u, \FirstTick u),\, (\Phase v, \FirstTick v)]$

	if not $\Done1 u$ then #\comment{Stage 1: Fast Leader Election}#
		$\text{\protocol{FastLeaderElection}}\,[(\Leader u, \Done1 u),\, (\Leader v, \Done1 v)]$

	if $\Done1 u$ and not $\Done2 u$ then #\comment{Stage 2: Approximation Stage}#
		$\text{\protocol{ApproximationStage}}\,[(\K u, \ell_u, \Done2 u),\, (\K v, \ell_v, \Done2 v)]$

	if $\Done1 u$ and $\Done2 u$ then #\comment{Stage 3: Refinement Stage}#
		$\text{\protocol{RefinementStage}}\,[(\K u, \L u),\, (\K v, \L v)]$
\end{lstalgo}
\begin{figure*}[p]
\framebox[\textwidth]{
\scalebox{0.8}{
\begin{minipage}{\textwidth}
\centering 
\begin{equation*}
\setlength\arraycolsep{0pt}
\begin{array}{lccccccc}
\text{state of node $u$:}\quad & \Big( 
\matrixbracex{\Level u,}{Junta Process} &
&
\matrixbracex{\Phase u,~ \FirstTick u, }{Phase Clock} &
&
\matrixbracex{\Leader u,~ \Done1 u, }{Leader Election} &
&
\matrixbracex{i_u,~ \K u,~ \L u,~ \Done2 u}{Approximation and Refinement Stages}\Big) \\
\text{state space:} &
\mathbb{N}_0 &
{}\times{}&
\mathbb{N}_0 \times \set{0,1} &
{}\times{}&
\set{0,1}\times\set{0,1} &
{}\times{}&
\mathbb{N}_0 \times \set{-1, 0,1, \dots} \times \mathbb{N}_0 \times \set{0,1}
\end{array}
\end{equation*}
\end{minipage}
}
}
\caption{State space of protocol \protocol{CountExact} (\cref{alg:exact})}\label{fig:exact-state-space}
\end{figure*}

\begin{lstalgo}[p]{\protocol{CountExact} Approximation Stage\label{alg:apx-stage}}
Interaction $[(\K u, \ell_u, \Done2 u),\, (\K v, \ell_v, \Done2 v)]$ in the #\protocol{ApproximationStage}#protocol:
	if $\FirstTick u = \ttrue$ then
		if $\Leader u = \ttrue$ and $i_u = 0$ then#\comment{initialize first phase}#
			$\ell_u \gets 1$

		if $\Leader u$ and $\ell_u \geq 4$ then#\comment{found an approximation}#
			$\Done2 u \gets \ttrue$#\label{ln:apx-decision}#
			$\K u \gets i_u \cdot 2^{\Level u-8} - \floor*{\log{\ell_u}}$

		$(i_u, \ell_u) \gets \left(i_u + 1,~ \ell_u \cdot 2^{2^{\Level u-8}}\right)$ #\label{ln:apx-load-explosion}\comment{start new phase: load explosion}#

	$\displaystyle \left(\ell_u, \ell_v\right) \gets \left( \floor*{(\ell_u + \ell_v)/2}, \ceil*{(\ell_u + \ell_v)/2} \right)$ #\label{ln:apx-load-balancing}\comment{classical load balancing}#

	$\Done2 u \gets \max\set{\Done2u, \Done2v}$ #\label{ln:apx-broadcast}\comment{broadcast $\Done2{}$}#
\end{lstalgo}
\begin{lstalgo}[p]{\protocol{CountExact} Refinement Stage \label{alg:refinement-stage}}
Interaction $[(\K u, \L u),\, (\K v, \L v)]$ in the #\protocol{RefinementStage}#protocol:
	if $\Phase u = 0$ then#\comment{initialize agents and broadcast $\K u$}#
		$\left(k_u, k_v, \ell_u, \ell_v \right) \gets \left(\max\set{k_u, k_v},\, \max\set{k_u, k_v},\, 0,\, 0\right)$#\label{ln:refinement-init}#

	if $\FirstTick u = \ttrue$ then
		if $\Phase u = 1$ and $\Leader u = \ttrue$ then#\label{ln:refinement-leader-load}\comment{the leader starts with load $2^8\cdot2^{\K u}$}#
			$\displaystyle \ell_u \gets 2^8 \cdot 2^{k_u}$

		if $\Phase u = 2$ then#\comment{multiply load with $2^{\K u}$}\label{ln:refinement-multiply}#
			$\displaystyle \ell_u \gets \ell_u \cdot 2^{k_u}$

$\displaystyle \left(\ell_u, \ell_v\right) \gets \left( \floor*{(\ell_u + \ell_v)/2}, \ceil*{(\ell_u + \ell_v)/2} \right)$
\end{lstalgo}%

\subsection{Fast Approximation}\label{sec:apx-stage}
In this section we assume that a unique leader has been elected and that all
agents are synchronized via the phase clocks. Additionally, we assume that all
agents are on the same maximal junta level $\Level*$ w.r.t.\ the modified junta
process from \cite{DBLP:journals/corr/abs-1805-04586}.

The goal of protocol \protocol{ApproximationStage} is to compute $\log{n}$ up
to an additive error of $\pm 3$. The protocol starts with injecting
$2^{2^{\Level u-8}}= n^{\eta}$ tokens into the system for some $0 < \eta \leq
1$. It alternates between increasing the number of injected tokens by a factor
of $n^{\eta}$ and load balancing until the total number of tokens $M$ is at
least $n/2$. Unfortunately it is possible that $M$ is very close to
$n^{1+\eta}$, resulting in a multiplicative error. Hence, the protocol outputs
$\K u=\log M-\floor*{\log \L u}$. In \cref{lem:apx-stage} we will show
that $\log n-3\leq k_u\leq \log n+3$.

The interactions of the protocol are defined in \cref{alg:apx-stage}. The
protocol runs in multiple phases of $\BigO{n \log{n}}$ interactions each (the
phases are determined by the phase clock). Every agent $v$ uses the variables
$i_v$ (initialized to $0$) as a phase counter and $\ell_v$ for load balancing.
The leader $u$ has an additional variable $\K u$ in which it eventually stores
the approximation. In the first phase, every agent $v$ initializes its load
$\ell_{v}$ to $0$ (non-leaders) or to $1$ (leader). In the first interaction of
every phase, every agent $v$ increases the phase counter $i_v$ and multiplies
its load with $n^\eta$ (see \cref{ln:apx-load-explosion}). During the remainder
of the phase, all agents us the classical load balancing process from
\cite{BFKK19} to balance their tokens. The leader $u$ additionally checks
before the multiplication whether it has a load of at least $4$ (in which case
the total load is at least $2n$ \whp/). If this is the case, the leader
calculates the approximation $\K u$ as $i_u \cdot \eta -
\floor*{\log{\ell_u}}$ and raises the flag $\Done2u$, indicating that the
approximation stage has concluded (see \cref{ln:apx-decision}). The flag is
then sent to all other agents via one-way epidemics (see
\cref{ln:apx-broadcast}), and the raised flag terminates the process.
\begin{lemma} \label{lem:apx-stage}
Let $u$ be the leader. After at most $\BigO{n\log{n}}$ interactions of the
Approximation Stage all agents $v$ set $\Done2 v$ to $\ttrue$ \whp/.  At that
time, $\K u = \log{n} \pm 3$ \whp/. The protocol uses \whp/ at most
$\BigOTilde{n}$ states.%
\end{lemma}

\begin{proof}
\def\hati{\hat\iota}
First we bound the number of phases until all agents $v$ have $\Done2v =
\ttrue$. When the leader $u$ sets $\Done2u$ to \ttrue, all other agents follow
via one-way epidemics in the same phase, \whp/. Therefore it suffices to bound
the number of phases until the leader $u$ sets $\Done2u$ to \ttrue.

Let $\hati$ be the first phase in which the leader has a load of at least
$4$. From load balancing \cite{BFKK19} it follows that the total load $M$ at
that time is at least $2n$. At the beginning of Phase 1 the leader $u$ has
$n^{\eta}=2^{2^{\Level v-8}}$ tokens and all other agents are empty. In the
first interaction of Phase $i>1$ all agents $v$ multiply their load with
$n^{\eta}$ (see \cref{ln:apx-load-explosion}). Hence, in Phase $i\leq\hati$
the total load is $M = n^{i\eta}$, \whp/ and in Phase $\hati$ we have $2n
\leq n^{\hati\eta} \leq 6\cdot n^{1+\eta}$ (see \cite{BFKK19}).
From \cref{lem:junta} we obtain that \whp/ $\log\log{n} - 4 \leq \Level v$ and
therefore $\eta \geq 1/2^{12}$ \whp/. Since $n^{\hati\eta} \leq 6\cdot n^{1+\eta}$
we get $\hati \leq 1/\eta + 1 + \LittleO{1}$. Hence the number of phases $\hati$
until $u$ sets $\Done2u$ to \ttrue is bounded by a constant $\hati = \BigO{1}$.

We now bound the quality of the approximation $\K u$.
In Phase $\hati$, the leader sets its variable $\K u$ to $\K u = \hati \cdot 2^{\Level u - 8} - \floor*{\log \L u}$.
From load balancing \cite{BFKK19} we get $\L u \leq 2^{\hati \cdot 2^{\Level u - 8}} / n + 1.5$ and therefore
\allowdisplaybreaks
\begin{align*}
\K u &\geq \hati \cdot 2^{\Level u -8 } - \log\left(2^{\hati \cdot 2^{\Level u - 8}}/n + 1.5\right) \\
& \numberthis\label{eq:quality-of-approximation}\geq \hati \cdot 2^{\Level u - 8} - \log{\left(2^{\hati\cdot2^{\Level u - 8}}/n \cdot \left(1 + 1.5/2\right)\right)} \\
& = \hati \cdot 2^{\Level u - 8} - \log{\left(2^{\hati\cdot2^{\Level u - 8}}\right)} + \log{n} - \log{1.75} \\
& \geq
\log{n} - 3 \enspace ,
\intertext{where \eqref{eq:quality-of-approximation} holds since $2^{\hati\cdot 2^{\Level u - 8}} \geq 2n$.
Analogously,}
\K u
& \leq \hati \cdot 2^{\Level u - 8} - \log{\L u } + 1 \\
& \leq \hati \cdot 2^{\Level u - 8} - \log{\left(2^{\hati\cdot2^{\Level u - 8}}/n\cdot \left(1 - 1.5/2\right)\right)} + 1 \\
& \leq \log{n} + 3 \enspace .
\end{align*}

It remains to bound the required number of states. For each agent $v$, the
variables $i_v$ and $\K v$ require \whp/ at most $\BigO{\log{n}}$ many states,
and the variable $\ell_v$ stores up to $\BigO{n^\eta}\leq n$ many tokens.
\end{proof}

\subsection{Refining the Approximation}\label{sec:refinement-stage}
As before, we assume that a unique leader has been elected and that all
agents are synchronized via the phase clocks.
We also assume that the leader holds an approximation of $\log{n} \pm 3$.

In the Refinement Stage, every agent $v$ holds a variables $\K v$ for the
approximation of $\log{n}$ and a variable $\ell_v$ that is used for load
balancing. The protocol consists of 3 phases. The first phase is used for
initialization. The value $\K u$ (the approximation calculated by the leader in
the previous stage) is spread among all agents via one-way epidemics
(\cref{ln:refinement-init}). In the beginning of the second phase the leader
injects $2^8\cdot 2^{\K u}$ many tokens into the system which are balanced, as
before. In the beginning of third phase every agent multiplies its load with
$2^{\K u}$, and then the load is balanced again. In \cref{lem:refinement-stage}
we show that total number of tokens $M$ is bounded by $2^2\cdot n^2 \leq M \leq
2^{14}\cdot n^2$. At the end of the phase every agent can compute $n$ as its
output function $\omega(v) = \round{{2^8\cdot2^{2\K v}}/{\L v}}$.
\begin{lemma} \label{lem:refinement-stage}
After at most \BigO{n\log{n}} interactions of the Refinement Stage, all agents
$v$ output $\omega(v) = n$ \whp/. The protocol uses \whp/ at most
$\BigOTilde{n}$ states.
\end{lemma}
\begin{proof}%
At the end of Phase $0$, \whp/ all agents know $\K u$ from the leader according
to one-way epidemics. At the end of Phase $1$, \whp/ $C \cdot 2^{\K u}$ tokens
have been balanced such that $\L v = \BigTheta{1}$ for any agent $v$. At the
end of Phase $2$, a total of at least $M \geq 4 \cdot n^2$ tokens have been
balanced and every agent $v$ has load $\L v = \round{C\cdot2^{2\K u}/n} \pm 1$
\whp/.

We now need to show that every agent $v$ can output $n$ using its output
function $\omega(v) = \round{{C2^{2\K{v}}}/{\L{v}}}$, provided that $v$ knows
the value $\K v = \K u = \log{n} \pm 3$ from the leader and the total load $M$
is at least $M = C\cdot2^{2\K u} \geq 4n^2$ tokens. The proof is based on the
proof of Lemma~3.8 from \cite{DBLP:journals/corr/abs-1805-04832}.

Observe that after the load balancing, \whp/ \makebox{$\L v = C\cdot2^{2\K u}/n + r$}
where $r \in [-1.5,1.5]$ (the error term $r$ covers the remaining discrepancy
and the rounding error). Let $\hat\omega(v)$ be the output of agent
$v$ before the rounding is applied. We derive, analogously
to~\cite{DBLP:journals/corr/abs-1805-04832},
\begin{align*}
\hat\omega(v) &= \frac{M}{\L{v}}
= \frac{C\cdot2^{2\K u}}{        C\cdot2^{2\K u}/n + r}
= \frac{C\cdot2^{2\K u}\cdot n}{ C\cdot2^{2\K u} + rn } \\
& = \frac{C\cdot2^{2\K u}\cdot n}{ C\cdot2^{2\K u} \left( 1 + \frac{rn}{C\cdot2^{2\K u}} \right) }
= \frac{n}{                                             1 + \frac{rn}{C\cdot2^{2\K u}} }
\intertext{which gives us, since $r \in [-1.5, 1.5]$ }
\hat\omega(v) &\geq \frac{n}{1 + \frac{1.5n}{C\cdot2^{2\K u}}}
\stackrel{\left(M \geq 4n^2\right)}{\geq} \frac{n}{1 + \frac{1.5n}{4n^2}} \\
& = \frac{n}{1+\frac{1}{3n}} = n - \left(\frac{1}{3} - \frac{1}{3(3n+1)} \right) > n - \frac{1}{2} \\
\lefttag{\text{and }}
\hat\omega(v) &\leq \frac{n}{1 - \frac{1.5n}{C\cdot2^{2\K u}}}
\stackrel{\left(M \geq 4n^2\right)}{\leq} \frac{n}{1 - \frac{1.5n}{4n^2}}\\
& = \frac{n}{1-\frac{1}{3n}} = n + \frac{1}{3} + \frac{1}{3(3n-1)} < n + \frac{1}{2} \enspace .
\end{align*}
Therefore, $\omega(v) = \round{\hat\omega(v)} = n$.
\end{proof}

\subsection{Analysis of \protocol{CountExact}}

We now show \cref{thm:counting}. Due to space limitations, we only show that
\protocol{CountExact} calculates the value of $n$ \whp/. In
\fullversion we describe how to combine the protocol with a
\emph{slow} always-correct protocol, which shows that protocol
\protocol{CountExact} stabilizes.

\begin{proof}
From \cref{lem:junta} it follows that after $\BigO{n\log{n}}$
interactions the junta is elected and all agents are inactive. From that point
on no agent in the protocol is ever re-initialized again. Let $u$ be the single
leader that \whp/ concludes the \protocol{FastLeaderElection} protocol
according to \cref{lem:fast-leader-election} after at most $\BigO{n\log{n}}$
interactions. According to \cref{lem:apx-stage}, the leader $u$ computes
$\log{n} \pm 3$ and sets $\Done2u$ to $\ttrue$ after at most $\BigO{n\log{n}}$
further interactions \whp/. Finally, according to \cref{lem:refinement-stage},
all agents output $n$ after $\BigO{n\log{n}}$ further interactions. The total
number of interactions therefore is $\BigO{n\log{n}}$. Observe that all
protocols require at most $\BigOTilde{n}$ states. This concludes the proof.
\end{proof}

In order to argue that
\protocol{CountExact} stabilizes, we use the same idea as in
\cref{sec:stable-short}, where we combined the protocol for approximate
counting with a slow protocol which is always correct. Together with an error
detection mechanism, this gives a protocol that stabilizes \whp/ in
$\BigO{n\log{n}}$ interactions such that every agent outputs the exact value
$n$, using $\BigOTilde{n}$ states. Further details can be found in
\fullversion.

\bibliographystyle{plain}
\bibliography{paper}
\clearpage
\appendix

\section{Load Balancing} \label{apx:load-balancing}

In this appendix we show \cref{lem:powers-of-two-load-balancing}. The proof
follows along the lines of the proofs of Lemmas 2 and 3 from \cite{BFKK19}. We
re-state the lemma as follows.

\restateLoadBalancing*
\begin{proof}
In the following, we will argue about individual tokens. However, this is only
used for the sake of the analysis. In the actual process, all agents store at
all times only the logarithm of their current numbers of tokens. We assume that
each agent $u$ holds a number of $\ell_u(t) = 2^{k_u(t)}$ tokens which are
stored in some order. Initially, the tokens are ordered arbitrarily. We define
the height $h(b)$ of a token $b$ as the number of tokens below $b$ w.r.t.\ that
order. Whenever an agent $u$ interacts with another agent $v$ that does not have
any load, every second token is moved from agent $u$ to agent $v$ and stored on
$v$ in the original order.
\begin{figure}[h]
\centering
\includegraphics{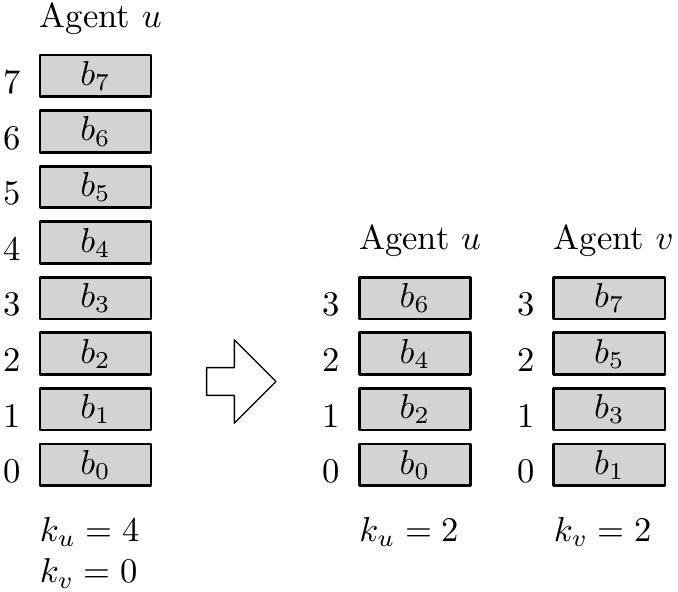}
\caption{Example: Agent $u$ moves every second token to the previously empty
agent $v$.} \label{fig:balancing-every-second}
\end{figure}
See \cref{fig:balancing-every-second} for an example. Therefore, all tokens
participating in the interaction reduce their height by a factor of at least
$2$. According to the assumptions of \cref{lem:powers-of-two-load-balancing}
we have $2^{\kappa} \leq 3/4 \cdot n$ and thus there are at most $3/4 \cdot n$
tokens in the system. Even when spread out as much as possible, these tokens
cannot take up more than $3/4\cdot n$ agents. Therefore, the probability that a
specific token $b$ of height $h(b) > 0$ reduces its height by a factor of $2$
when participating in an interaction with an empty agent is at least $1/4$.
(Note that the actual probability is bounded only by $1/4 \cdot (1-1/n)$,
since the leader does not participate in the load balancing process. For
simplicity we work with the expression $1/4$, and for $n$ large enough our
final results holds nonetheless.) We now observe that
\begin{enumerate}
\item the agent on which $b$ is stored is selected with probability $2/n$, and
\item the other agent selected does not have any load with probability $1/4$.
\end{enumerate}
Combining these two events gives us that in interaction $t$ with probability
$1/(2n)$ the token $b$ reduces its height by at least a constant factor of $2$ or
has already reached height $0$. We denote this event as a \emph{good
interaction $t$ for token $b$}. We now consider $16 n \log{n}$ interactions and
use $X_{b,i}$ to denote an indicator random variable for good interactions of
token $b$. Since for one specific token the interactions occur independently
from each other, we may apply Chernoff bounds to $X_b = \sum_{i=1}^{16 n
\log{n}}X_{b,i}$ and obtain
\begin{equation*}
\Probability{X_b < \log{n}-1} < \Exp{-\frac{49}{16}\log{n}} < n^{-3} \enspace.
\end{equation*}
From the union bound over all tokens we obtain that after at most $16 n \log{n}$
interactions all tokens have reached height $0$ with probability $1-n^{-2}$.
\end{proof}

\section{The Stable Protocol} \label{apx:stable}

In this appendix we show that we can build a stable protocol using only an additional
multiplicative factor of $\BigO{\log{n}}$ states. We first describe our error
detection mechanisms and then show \cref{thm:approx-2} of \cref{thm:approx}.
Finally, we show how to reduce the required number of states by accepting
a small number of agents that do not output the correct result.

Our approach is closely related to the idea of building a hybrid protocol in
\cite{GS18} (see also \cref{sec:auxiliary-protocols}). We need to handle two
possible sources of errors. The first type of error comes from leader election,
where it may happen that no agent ever completes the fast leader election
protocol, or multiple leaders do. The second type of error occurs when the
\Search Protocol fails.

If no agent completes leader election, then no agent will ever set $\Done1{}$. In
that case, we use the following solution: From the very beginning, all agents
run a slow backup protocol in parallel. (We give a description of the backup
protocol in \cref{apx:backup-protocol}.) Only if an agent $v$ completes the
leader election protocol and sets $\Done1{v}$ to $\ttrue$, it stops the
execution of the slow backup protocol and proceeds with the \Search Protocol,
which is governed by the leader.

If multiple agents conclude the leader election stage as a leader, then this
error can be detected whenever two leaders interact with each other. The agent
$v$ which detected the error raises an error flag $\Error v$. This flag is
adopted by all other agents using one-way epidemics. When an agent sets its error
flag, it ignores all of its previous computations and executes a new instance
of the backup protocol. Note that stopping the backup protocol while running
the \Search Protocol allows us to save states by re-using the states from the
backup protocol (at costs of having to re-start the backup protocol in case of
a later error).

The second type of error occurs when some agents fail to properly execute the
\Search Protocol. E.g., either load balancing or maximum broadcast could fail,
resulting in an incorrect value $\K u$ at the leader $u$. In the
remainder of this section, we describe how such an error can be detected
after the execution of the \Search Protocol. As before, whenever an agent
detects an error, it raises the error flag and switches to a new instance of
the backup protocol.

\bigskip

To detect an error during the execution of the \Search Protocol, all agents use
after the \Search Protocol (i.e., when $\Done2{v} = \ttrue$ for agent $v$) an
additional error detection protocol \protocol{ErrorDetection} defined in
\cref{alg:error-detection}. This protocol is enabled by replacing the
Broadcasting Stage from the \protocol{Approximate} (\cref{alg:approx}) with the
Error Detection Stage by calling the \protocol{ErrorDetection} protocol
(\cref{alg:error-detection}) as follows:

\begin{lstalgo}[H]{Additional Lines for Protocol \protocol{Approximate} to enable error detection\label{alg:approx-with-error-detection}}
Updated lines of protocol #\protocol{Approximate}:\setcounter{lstnumber}{\value{lstnumberAlgApproximate}}\addtocounter{lstnumber}{-1}#

	if $\Done1 u$ and $\Done2 u$ then #\comment{Stage 3: Error Detection Stage}#
		$\text{\protocol{ErrorDetection}}\,[(\K u, \Done2 u), (\K v, \Done2 v)]$
\end{lstalgo}

The error detection protocol runs in $5$ phases. We assume that all agents start
counting the phases over from $0$ when they enter the Error Detection Stage,
and they stop counting phases when they have reached phase $4$. The protocol
requires at each agent $v$ in addition to the states from \cref{alg:search} a
flag $\Error v$ which is initially set to \tfalse and an additional variable
$\L v \in \set{0, 1, \dots, 32}$ used for load balancing. Intuitively, the
leader $u$ balances $2^{\K u}$ tokens among the agents, and the resulting load
at each agent allows the leader to validate its result. The phases of the
\protocol{ErrorDetection} protocol have the following purposes.
\begin{enumerate}
\item
In the first phase ($\Phase u = 0$), the leader $u$ initializes another agent
in its first interaction with a load of $2^{\K{u}-2}$ tokens stored in powers of two
(\cref{ln:load-error-detection}). Note that the $-2$ in the exponent is required
that the powers-of-two load balancing process perfectly balances all tokens \whp/,
and it will be compensated in the third
phase. All non-leader agents set their load to $-1$ once they enter the Error
Detection stage (\cref{ln:enter-error-detection}).
\item
In the second phase, all agents $v$ except the leader perform the powers-of-two load
balancing procedure on their $\K{v}$ values.
\item
In the third phase, every agent $v$ initializes its additional token counter $\L
v$ either to $0$ if $\K v = -1$ (the agent did not have any tokens in the
powers-of-two load balancing process) or to $32$ if $\K v = 0$ (the agent had one
token in the powers-of-two load balancing process). Here, the factor of $32$
over-compensates for the smaller number of tokens initially used in the first
phase.
\item In the fourth phase, all agents $v$ perform the classical load balancing process
from \cite{BFKK19} on their $\L{v}$-values.
\item In the last phase, all agents perform the following. Every agent $v$ checks if
its remaining load $\L v$ is at least $3$. If this is not the case, an error has
occurred, the agent raises the error flag $\Error v$ and moves over to the backup protocol.
Furthermore, every agent $v$ checks whether the remaining discrepancy is at most $2$.
In addition, once an agent $v$ has reached $\Phase v = 4$, it stops its phase
clock. In that way, it runs through the error detection protocol exactly
once and then remains in Phase~4 forever (unless an error is detected).
\end{enumerate}

The output of every agent $v$ is the value $\K u$ from the leader $u$, which is set in \cref{ln:set-output-error-detection}. It is then sent to all agents in Phase~4 via one-way epidemics.

\begin{lstalgo}[H]{Error Detection protocol in protocol \protocol{Approximate}\label{alg:error-detection}}
Interaction $(u, v)$ in protocol #\protocol{ErrorDetection}:#

	if $\Done2 v = \tfalse$ then #\comment{a new agent enters error detection}#
		$\left(\K v,\, \L v,\, \Done2 v\right) \gets (-1,\, 0,\, \ttrue)$#\label{ln:enter-error-detection}#

	if $\Phase u' = 0$ and $\FirstTick u = \ttrue$ and $\Leader u = \ttrue$ then
		$\K v \gets \K u - 2$ #\comment{load infusion}\label{ln:load-error-detection}#

	if $\Phase u' = 1$ and $\Leader u = \Leader v = \tfalse$ then #\comment{powers-of-two load balancing}#
		if $\min\set{\K u, \K v} = -1$ and $\max\set{\K u, \K v}>0$ then 
			$\K u, \K v \gets \max\set{\K u, \K v} - 1$ 

	if $\Phase u' = 2$ and $\FirstTick u = \ttrue$ then #\comment{initialize load balancing with constantly many tokens per agent}#
		if $\K u = -1$ or $\Leader u = \ttrue$ then
			$\L u \gets 0$
		else if $\K u = 0$ then
			$\L u \gets 32$
		else
			$\Error u \gets \ttrue$

	if $\Phase u' = 3$ then #\comment{classical load balancing}#
		$\left(\L u, \L v\right) \gets \left( \floor*{(\L u + \L v)/2} , \ceil*{(\L u + \L v) / 2} \right)$
	
	if $\Phase u' = 4$ then
		if $\Leader u = \ttrue$ and $\FirstTick u = \ttrue$ then
			$\K u \gets \round{\K u + 3 - \log{\L u}}$#\label{ln:set-output-error-detection}\comment{compute approximation of $\log{n}$}#
		if $\L u < 3$ or $\abs{\L u - \L v} > 2$ then
			$\Error u \gets \ttrue$ #\label{ln:balancing-error}\comment{balancing error detected}#
		$\K u \gets \max\set{\K u, \K v}$ #\comment{broadcast result from the leader}#
		stop the phase clock.
\end{lstalgo}

We now put everything together and show the stability of \cref{alg:approx}.
Recall that in \cref{sec:analysis-search-protocol} we have shown that the protocol works as claimed \whp/. It remains to show that the protocol stabilizes.

\begin{proof}[Proof of \cref{thm:approx-2} and \cref{thm:approx-3} of \cref{thm:approx}]

Recall that in any error case, all agents $v$ raise their $\Error{v}$ flag. The
$\Error{v}$ flag is always adopted by all other agents using one-way epidemics.
Whenever an agent $v$ sets its $\Error{v}$ flag, it ignores all of its previous
computations and starts to execute a new instance of the backup protocol. The
backup protocol stabilizes after $\BigO{n\log^2{n}}$ interactions \whp/ (see
\cref{lem:backup-protocol}) in \cref{apx:backup-protocol}.

We first show that the protocol stabilizes unless an error is detected.

\Whp/, the phase clocks work correctly such that the phase length is
sufficiently long enough for load balancing to complete (see
\cref{lem:phase-clocks}). The powers-of-two load balancing process balances
then the load in Phase $1$ such that all agents $v$ (except the leader $u$) have a
value $\K v \in \set{-1, 0}$ (see \cref{lem:powers-of-two-load-balancing}). The
resulting number of tokens is then multiplied with 32 and balanced in Phase 3 
using the classical load balancing process from \cite{BFKK19}.

From the result on load balancing we obtain the following.
\begin{itemize}
\item
From \cref{lem:search} we obtain that $\K u \geq \floor*{\log{n}}$ \whp/. Therefore,
at least $2^{\K u - 2} \cdot 32 > 4 \cdot n $ tokens are balanced over $n$
agents. From \cite[Theorem 1]{BFKK19} we obtain that after balancing every agent
has load at least $3$ \whp/. Otherwise, the error flag is raised by every agent
which has load less than $3$ (see \cref{ln:balancing-error}).

\item
As before,
the error flag is raised by every agent which observes a balancing error.
However, note that after balancing the discrepancy is at most $2$ \whp/
\cite{BFKK19}. Therefore, a load balancing error does not occur \whp/ and
thus in \cref{ln:balancing-error} of \cref{alg:error-detection} the error flag
remains unset \whp/.
\end{itemize}

Once the load balancing has successfully completed such that $\L v \geq
3$ at all agents $v$, for all pairs of agents $(u, v)$ we have $\abs{\L u - \L v} \leq 2$, all
agents are in Phase~4, and all agents have adopted the maximal value $\K u$ from the leader $u$,
the protocol does not provide any transitions due to which an agent could change its
output. Together this implies that the protocol stabilizes \whp/ if $\K u \in
\set{ \floor*{\log{n}}, \ceil*{\log{n}}}$ after the \Search Protocol.

Observe that the leader $u$ computes its approximation in the
\protocol{ErrorDetection} protocol as $ \K u = \round{\K u + 3 - \log{\L u}}$
(see \cref{ln:set-output-error-detection}). We use $\L u = \round{32\cdot2^{\K
u}}\pm1$ and get that $\K u = \log{n} \pm \delta$ for some $\delta < 1$.
Therefore, $\K u \in \set{\floor*{\log{n}}, \ceil*{\log{n}}}$ provided the load
balancing has concluded with a maximal remaining discrepancy of $2$.

To now show \cref{thm:approx-2} of \cref{thm:approx} we handle the various sources
of errors one after the other.
\begin{itemize}
\item Leader Election Errors:\par
If no leader concludes the leader election protocol from \cite{GS18}, then
no agent ever sets $\Done1{}$ to \ttrue. In that case all agents output the
result from the backup protocol indefinitely.

If multiple agents conclude the leader election protocol, then
this is detected when two of them interact with each other. In that case,
both agents raise the error flag, causing the system to stabilize
due to the backup protocol.

\item Synchronization Errors:\par
Observe that the $\Done2{}$ flag is transmitted to all agents via one-way
epidemics (\cref{ln:enter-error-detection} of \cref{alg:error-detection}). Once
the flag is set, the \protocol{ErrorDetection} protocol runs in 5 phases
counted in the $\Phase{}'$ variable. It is crucial that this number of phases is a
constant, and that the phase clock \emph{gets stuck} in Phase $4$.

Recall that we define that the phase clocks are updated at the beginning of an
interaction. If two agents interact thereafter that have different $\Phase{}'$
values, the execution has become asynchronous and the participating agents raise
the $\Error{}$ flag. Note that this also implies that every agent initiates at
least one interaction of every phase (thereby executing the special
instructions conditioned on $\FirstTick{} = \ttrue$), since otherwise a phase
difference of at least 1 is detected in a later interaction.

\item Errors in the \Search Protocol:\par
During the \Search Protocol, no explicit error detection mechanisms are used.
However, observe that as soon as a leader is present the \Search Protocol, the
total load in the system starts to grow over time. This follows from the fact
that the phase clock performs an infinite sequence of ticks with probability
$1$. Therefore the value $\K u$ from the leader is incremented repeatedly, and
therefore the total load in the system grows such that the leader will
eventually conclude the \Search Protocol.

At the time when the leader $u$ concludes the \Search Protocol, its estimate
$\K u$ of $\log{n}$ might be too small. However, this error is detected in the
\protocol{ErrorDetection} protocol and the error flag is raised by some agent.
As soon as the error flag is raised, all agents adopt the error flag via
one-way epidemics, and the protocol then stabilizes via the backup protocol.
\end{itemize}

It remains to show that the protocol stabilizes \whp/ within
$\BigO{n\log^2{n}}$ interactions. Recall that without the
\protocol{ErrorDetection} protocol the protocol converges \whp/ within
$\BigO{n\log^2{n}}$ interactions (see \cref{thm:approx-1} of
\cref{thm:approx}). The \protocol{ErrorDetection} protocol runs in $5$
additional phases. Together this gives us that \whp/ the protocol stabilizes in
$\BigO{n\log^2{n}}$ interactions.

Regarding the required number of states, we observe that the
\protocol{ErrorDetection} protocol operates on the same states as the \Search
Protocol, and the only non-constant size variable of any agent $u$ is $\K u$,
which uses \whp/ up to $\log{n}$ states -- as in the \Search Protocol. Except
for the Backup Protocol, the stable protocol therefore does not require more
states then the \Search Protocol. The Backup Protocol requires
$\BigO{\log^2{n}}$ states (see \cref{apx:backup-protocol}. Since it has to be
run in parallel until an agent $u$ sets $\Done1u$ to \ttrue, our protocol
requires in total a number of $\BigO{\log^2{n}\log\log{n}}$ states \whp/.

It remains to show that we can further reduce the required number of states if
we slightly relax the output definition such that only all but $\log{n}$ agents
output the correct result. This is achieved by using a slightly modified backup
protocol which requires only \BigO{\log{n}} states, but up to $\log{n}$ agents
do not output the correct result.
\end{proof}

\section{Backup Protocols} \label{apx:backup-protocol}

\subsection{Backup Protocol for Approximate Counting}
\def\kmax_#1{k^{\text{\normalfont max}}_{#1}}
\def\k_#1{k^{\phantom{\text{\normalfont i}}}_{#1}}
In this appendix we define our backup protocol to compute $\floor*{\log{n}}$,
using at most $\log^2{n}$ states.

In the slow but stable backup protocol, each agent starts with one token. Agents
then keep combining tokens in groups of sizes which are powers of two. Each
agent $v$ keeps a pair $(\k_v, \kmax_v)$, where $\k_v$ is the number of tokens
in this agent and $\kmax_v$ is the largest $\k_w$ that agent $v$ is aware of. Let
$\k_v = -1$ represent an empty agent which does not have any tokens. The initial
state of an agent is $(0, 0)$, that is, initially each agent has one token. During
an interaction of two agents $(u, v)$, there are now two possible cases. If both
agents have the same number of tokens, than the first agent takes all of them.
Otherwise, if the two agents have different numbers of tokens, no exchange of
tokens is possible. In both cases, the agents always update the maximum number
of tokens, which is stored in the $\kmax_{}$-values. Formally, the protocol has
the transitions
\begin{align}
\delta \left[ \left(\k_u, \kmax_u \right),\left(\k_v, \kmax_v \right) \right] & =
\begin{cases}
	\left[ \left(\k_u+1, K \right),\left(-1, K \right) \right] & \text{ if } \k_u = \k_v \geq 0\\
	\left[ \left(\k_u, K \right),\left(\k_v, K \right)\right] & \text{ otherwise, }
\end{cases} \\
\lefttag{\text{where}} K &= \max\set{\kmax_u, \kmax_v} \enspace .
\end{align}
We now show the following result for the backup protocol.
\begin{lemma} \label{lem:backup-protocol}
Let $K_i = \set{u : \k_u = i }$ be the set of agents which have $\k_u = i$ and
let $n_i$ be the $i$-th bit of the binary representation of $n$. Then the
backup protocol converges to a configuration where
\begin{itemize}
\item $\abs{K_i} = n_i$ for all $i \geq 0$ and
\item $\max_{u}\set{\k_u} = \floor{\log{n}}$ and $\kmax_v = \floor{\log{n}}$
for all agents $v$.
\end{itemize}
It uses at most $(\log{n}+1)^2$ states and stabilizes \whp/ within
$\BigO{n^2\log^2{n}}$ interactions.
\end{lemma}
The proof follows straightforwardly from the coupon collector's problem. For
completeness, we formally show the claim as follows.

\begin{proof}
We show by induction on $j$ in $\set{0,1, \dots, \floor*{\log{n}}}$ that after
$j \cdot \BigO{n\log{n}}$ interactions for all $i < j$ it holds that $\abs{K_i}
= n_i$ \whp/.

The base case trivially holds. For the induction step, fix a $j \in \set{0,1,
\dots, \floor*{\log{n}}}$. From the induction hypothesis we obtain that for any
level $i < j$ there are $n_i$ token \emph{stuck} on that level. Note that $n_i
\leq 1$: this means that no further tokens from any agent in $K_i$ with $i < j$
can be merged, and therefore $\abs{K_j}$ cannot increase \whp/. We now apply a
standard coupon collecting argument: After $\BigO{n^2 \log{n}}$ interactions,
as many tokens as possible on level $j$ have been merged \whp/. Thus,
$\abs{K_j}$ is either $0$ or $1$ after $j \cdot \BigO{n\log{n}}$ interactions
\whp/. It remains to argue that $\abs{K_j} = n_j$.

From the binary representation of $n$ it follows that the total number of
tokens that can ever reach level $j$ is $\hat{k}_j = \left(n - \sum_{i=0}^{j-1}
n_j\cdot2^{i}\right) / 2^{j} = \floor*{n / 2^{j}}$. If $n_j = 0$, $\hat{k}_j$
is even. All $\hat{k}_j$ tokens can therefore be merged to level $j$ such that
eventually $\abs{K_j} = 0$. If $n_j = 1$, a single token remains on level $j$
indefinitely. This shows that $\abs{K_j} = n_j$ after $j \cdot \BigO{n^2
\log{n}}$ interactions.

From the binary representation of $n$ it furthermore follows that eventually
$\abs{K_{\floor*{\log{n}}}} = 1$. Let $u \in K_{\floor*{\log{n}}}$ at some time
$t$ and observe that $\k_{u} = \floor*{\log{n}} > \k_{v}$ for any $v \neq u$.
According to \cref{lem:broadcast} on one-way epidemics it follows that
$\kmax_{u} = \k_{u}$ is transmitted to all other agents after additional
$\BigO{n \log{n}}$ interactions after time $t$ \whp/.

Finally, the protocol requires at each agent at most $\floor*{\log{n}}+1$ states
for both, the $\k_u$ and the $\kmax_u$ values. This concludes the proof. 
\end{proof}

\subsection{Backup Protocol for Exact Counting}
\def\Counted#1{c_{#1}}
\def\Max#1{n_{#1}}
In the backup protocol for exact counting, every agent $u$ has variables $(\Counted u, \Max u)$
initially set to $(\tfalse, 1)$. Interactions are defined as
\begin{equation}\label{eq:backup-protocol-exact}
[(\Counted u', \Max u'),\, (\Counted v', \Max v')] = \begin{cases}
[(\tfalse, \Max u + \Max v),\, (\ttrue, \Max u + \Max v)] & \text{if } \Counted u = \Counted v = \tfalse \\
[(\Counted u, \max\set{\Max u, \Max v},\, (\Counted v, \max\set{\Max u, \Max v})] & \text{otherwise.}
\end{cases}
\end{equation}
Intuitively, in this protocol every agent starts with one token and a bit
$\Counted u$ indicating whether its token has been counted. Agents combine
tokens until eventually a single agent holds all $n$ tokens. Those agents which
have already been counted broadcast the maximum value they have observed so
far. The following lemma follows directly from coupon collecting.
\begin{lemma}
The backup protocol defined via \cref{eq:backup-protocol-exact} stabilizes
\whp/ within $\BigO{n^2\log{n}}$ interactions. When it stabilizes, every agent
outputs the exact value $n$.
\end{lemma}

\section{Fast Leader Election using $n^\epsilon$ States} \label{apx:leader-election}
\def\Coins#1{\ell_{#1}}
The following protocol performs leader election in $\BigO{n\log{n}}$
interactions using $\BigO{n^\epsilon}$ many states for some small constant
$\epsilon > 0$. It is a special case of the protocol presented in
\cite{DBLP:journals/corr/abs-1805-04586}.
It has the property that there is always at least one leader,
and when the first agent sets $\Done1 u$ to \ttrue, there is at most one leader
\whp/.

To simplify the description of the protocol, we assume that agents can toss a
random coin and generate one random bit uniformly and independently in each
interaction. While this is not covered by the actual population model, the
randomness from the scheduler can be exploited. This technique has been called
\emph{synthetic coins} by Alistarh et al.\
\cite{DBLP:conf/soda/AlistarhAEGR17}, and a simple analysis has been presented
in \cite{DBLP:conf/soda/BerenbrinkKKO18}. Intuitively, the idea is that every
agent keeps track of the parity of its current interaction number by flipping a
bit in every interaction. This bit then functions as a source of randomness for
the interaction partner.

The \protocol{FastLeaderElection} protocol operates for each agent $u$ on the
variables $\Leader u$, $\Coins{u}$ and $\Counter{u}$. The variable $\Leader u$
indicates whether $u$ is still a leader contender. When the protocol has
reached $\Phase u = 16$, it concludes by setting $\Done1 u$ to \ttrue. To get
an estimate of $\log{n}$, it uses the $\Level u$ variable from the
\protocol{Junta} process.

The protocol then runs in multiple phases. In even phases, all leader contender
agents sample $\BigTheta{\log{n}}$ random bits. In odd phases, the maximum of
the numbers represented by these random bits is sent to all agents via one-way
epidemics. If a leader contender $u$ observes a larger value in the broadcast,
it becomes a follower by setting $\Leader u$ to \tfalse. The protocol is
specified in \cref{alg:leader-election}

\begin{lstalgo}[H]{Protocol \protocol{FastLeaderElection} used in \cref{alg:exact}\label{alg:leader-election}}
Interaction $[(\Leader u, \Done1 u),\, (\Leader v, \Done1 v)]$ in the #\protocol{FastLeaderElection}# protocol:

	if $\Phase u = \Phase v$ then

		if $\Phase u$ is even and $\FirstTick u = \ttrue$ then #\comment{initialize a new round}#
			$\Coins{u}, \Counter{u} \gets 0$

		if $\Phase u$ is even and $\Leader u = \ttrue$ and $\Counter u < 2^{\Level u-8}$ then #\lcomment{generate random number}#
			$\Coins{u} [\Counter{u}] \gets \begin{cases} 0 & \text{with probability } 1/2 \\ 1 & \text{with probability } 1/2 \end{cases}$ #\comment{set $\Counter u$-th bit of $\Coins u$ to random value}#
			$\Counter u \gets \Counter u + 1$

		if $\Phase u$ is odd and $\Coins u < \Coins v$ then #\comment{$u$ observed another leader}#
			$\Leader u \gets \tfalse$#\label{ln:become-follower}#
			$\Coins u \gets \Coins v$ #\comment{adopt the maximum value}#

		if $\Phase u = 2^{13}$ then
			$\Done1 u \gets \ttrue$
\end{lstalgo}

We now show \cref{lem:fast-leader-election} from \cref{sec:auxiliary-protocols}.

\restateFastLeaderElection*

\begin{proof}
First, observe that the protocol runs in constantly many phases. Therefore,
every agent sets its flag $\Done1{}$ to \ttrue after at most \BigO{n\log{n}}
interactions \whp/.

Secondly, observe that there is always at least one leader contender. For every
phase $t$, let $L(t)$ be the set of agents $u$ which have sampled the maximum
value $\Coins u$. By definition, $L(t)$ is non-empty. Any contender can only
become a follower in \cref{ln:become-follower}, and due to that rule it is
impossible that any agent in the set $L(t)$ becomes a follower. Assume towards a
contradiction that an agent $w$ in $L(t)$ becomes a follower. This can only occur
if $\Coins w < \Coins v$ in some interaction with another agent $v$. However,
this is a contradiction to $w \in L(t)$, the set of agents $u$ with largest
value $\Coins u$.

Now consider an arbitrary but fixed pair of agents $(u, v)$. If $u$ and $v$ do
not sample precisely the same bits in every even phase, at least one of them
will become a follower in the next odd phase. The probability $p$ that both
agents sample the same sequence of bits during the entire course of the protocol
is
\[ 
p = 1/2^{2^{13}\cdot 2^{\Level u - 8}} \enspace .
\]
For the maximal junta level $\Level*$ in the modified junta process it
holds that $\Level* \geq \log\log{n} - 8$ \whp/, see \cref{lem:junta}.
We therefore get that
\[
p \leq 1/2^{2^{13}\cdot 2^{\log\log{n} -8}} = 1/n^2 \enspace .
\]
We now use the above observation that there is always at least one leader
contender. Let $u$ be a leader contender that concludes the
\protocol{FastLeaderElection} protocol by setting $\Done1{u}$ to \ttrue. We
take the union bound over all other agents $v \neq u$ and obtain that with
probably at most $1/n$ there exists a second leader contender.
\end{proof}

\section{Omitted Proofs from \cref{sec:counting}} \label{apx:omitted-proofs}

\begin{proof}[Proof of \cref{lem:refinement-stage}]
At the end of Phase $0$, \whp/ all agents know $\K u$ from the leader according
to one-way epidemics. At the end of Phase $1$, \whp/ $C \cdot 2^{\K u}$ tokens
have been balanced such that $\L v = \BigTheta{1}$ for any agent $v$. At the
end of Phase $2$, a total of at least $M \geq 4 \cdot n^2$ tokens have been
balanced and every agent $v$ has load $\L v = \round*{C\cdot2^{2\K u}/n} \pm 1$
\whp/.

We now need to show that every agent $v$ can output $n$ using its output
function $\omega(v) = \round*{{C2^{2\K{v}}}/{\L{v}}}$, provided that $v$ knows
the value $\K v = \K u = \log{n} \pm 3$ from the leader and the total load $M$
is at least $M = C\cdot2^{2\K u} \geq 4n^2$ tokens. The proof is based on the
proof of Lemma~3.8 from \cite{DBLP:journals/corr/abs-1805-04832}.

Observe that after the load balancing, \whp/ $\L v = C\cdot2^{2\K u}/n + r$
where $r \in [-1.5,1.5]$ (the error term $r$ covers the remaining discrepancy
plus $0.5$ for the rounding error). Let $\hat\omega(v)$ be the output of agent
$v$ before the rounding is applied. We derive, analogously to
\cite{DBLP:journals/corr/abs-1805-04832},
\begin{align*}
\hat\omega(v) &= \frac{M}{\L{v}}
= \frac{C\cdot2^{2\K u}}{        C\cdot2^{2\K u}/n + r}
= \frac{C\cdot2^{2\K u}\cdot n}{ C\cdot2^{2\K u} + rn }
= \frac{C\cdot2^{2\K u}\cdot n}{ C\cdot2^{2\K u} \left( 1 + \frac{rn}{C\cdot2^{2\K u}} \right) }
= \frac{n}{                                             1 + \frac{rn}{C\cdot2^{2\K u}} }
\intertext{which gives us, since $r \in [-1.5, 1.5]$ }
\hat\omega(v) &\geq \frac{n}{1 + \frac{1.5n}{C\cdot2^{2\K u}}}
\stackrel{\left(M \geq 4n^2\right)}{\geq} \frac{n}{1 + \frac{1.5n}{4n^2}}
= \frac{n}{1+\frac{1}{3n}} = n - \left(\frac{1}{3} - \frac{1}{3(3n+1)} \right) > n - \frac{1}{2} \\
\lefttag{\text{and }}
\hat\omega(v) &\leq \frac{n}{1 - \frac{1.5n}{C\cdot2^{2\K u}}}
\stackrel{\left(M \geq 4n^2\right)}{\leq} \frac{n}{1 - \frac{1.5n}{4n^2}}
= \frac{n}{1-\frac{1}{3n}} = n + \frac{1}{3} + \frac{1}{3(3n-1)} < n + \frac{1}{2} \enspace .
\end{align*}
Therefore, $\omega(v) = \round*{\hat\omega(v)} = n$.
\end{proof}

\section{Stable Exact Counting} \label{apx:stable-exact}

In this appendix we describe how to build a stable protocol to count 
the exact population size. 
We use a similar approach as in \cref{sec:approx}. We check for
errors during the course of the protocol, and if we detect an error, the
interacting agents raise an error flag. This error flag is sent to all other
agents via one-way epidemics. Every agent which has set the error flag performs a
simple backup protocol described in \cref{apx:backup-protocol} that outputs the
correct result with probability $1$ after $\BigO{n^2\log{n}}$ interactions
\whp/.

We first describe which error detection mechanisms are additionally required
for \cref{alg:exact} and then we show that the combined protocol stabilizes.

For the \protocol{FastLeaderElection} protocol (see \cref{apx:leader-election})
we know that there is always at least one leader. If, however, more than two
leaders conclude the leader election protocol, this error can be detected
whenever they directly interact with each other. In this case they both raise
the error flag.

The second type of error we detect during the course of the protocol are
synchronization problems with the phase clock. We assume that all agents count
the exact number of phases (until the last phase of the
\protocol{RefinementStage} protocol). When now two agents interact that have a
different phase counter, both agents raise the error flag.

Finally, in the \protocol{RefinementStage} protocol all agents $v$ verify in
\cref{ln:refinement-multiply} that they have a load of at least $\L v \geq 2^5$
before multiplying their load with $2^{\K v}$. At the same time, all agents start
to compare their $\K v$ values. If any agent does not have sufficiently many
tokens, or if two agents interact with different $\K v$ values, the agents
raise the error flag.

We now give the full proof for \cref{thm:counting}, which is re-stated for
completeness as follows.

{
\def\footnote#1{}
\restatethmcounting*
}

\begin{proof}
There is always at least one leader. If, however, multiple leaders conclude the \protocol{FastLeaderElection} protocol, this error is detected and the protocol
stabilizes due to the backup protocol.
We therefore consider in the following only the case when
precisely one leader concludes the \protocol{FastLeaderElection} protocol.

Note that In the way we elect the junta there is always at least one junta
agent. Therefore the phase clocks perform an infinite sequence of ticks with
probability $1$. However, we cannot guarantee that all agents are always in
exactly the same phase w.r.t.\ their ticks, and the length of the phases might
diverge drastically from the expectation (albeit only with low probability).
Now since all agents count the exact number of phases, whenever two agents
interact that have a different phase counter, both agents raise the error
flag and the protocol stabilizes due to the backup protocol.

The presence of a unique leader guarantees that in the \protocol{ApproximationStage}
protocol the total load is initially at least $1$. Recall that the phase clock
ticks with probability $1$, even if the agents are not properly synchronized.
Therefore, the load grows in the \protocol{ApproximationStage} protocol over
time such that eventually a leader $u$ concludes the protocol in
\cref{ln:apx-decision} by setting $\Done2 u $ to \ttrue.
Therefore, with probability $1$ the leader
starts the \protocol{RefinementStage} protocol by setting its load $\L u$ to
$C\cdot2^{\K u}$. In the \protocol{RefinementStage} protocol
all agents $v$ verify that they have a load of
at least $\L v \geq 2^5-1.5$ before multiplying their load with $2^{\K v}$.
If this
is the case for every agent, the value $\K u$ from the leader $u$ has been at
least $\log{n}-3$, and the total load is at least $4n^2$ as required in the proof
of \cref{lem:refinement-stage}.
Otherwise, the under-loaded agents raise the error flag
since the total load is insufficient to exactly compute $n$. Again, the protocol stabilizes due to the backup protocol.

If all agents have passed this last error check in Phase 2 of the
\protocol{RefinementStage} protocol (and all agents $v$ hold the same value $\K v$), then all agents will eventually output $n$ with
probability $1$. Otherwise, at least one agent $v$ is in the error state (or will enter the error state when comparing its $\K v$ value) and thus with probability $1$ all
agents will eventually output $n$ via the backup protocol.

It remains to argue the required number of states and the time until the protocol stabilizes.

For the required number of states, observe that for each agent $u$ \whp/ $\Level
u = \BigO{\log\log{n}}$ (see \cref{lem:junta}), $i_u = \BigO{1}$ and $\K u \leq
\log{n} + 3$ (see \cref{lem:apx-stage}), and thus $\L u$ holds at most
$\BigO{n}$ tokens in the \protocol{RefinementStage} protocol. The total number
of states is therefore \whp/ at most $\BigO{\log\log{n}\cdot \log{n} \cdot n}$.

The protocol stabilizes once all agents have concluded the second phase of the
\protocol{RefinementStage} protocol. The \protocol{FastLeaderElection} protocol
requires \whp/ at most $\BigO{n\log{n}}$ interactions (see
\cref{lem:fast-leader-election}), the \protocol{ApproximationStage} protocol also
requires \whp/ at most $\BigO{n\log{n}}$ interactions (see
\cref{lem:apx-stage}), and the \protocol{RefinementStage} protocol requires $3$
phases, i.e., $\BigO{n\log{n}}$ interactions, as well. Therefore the protocol
\protocol{CountExact} stabilizes within $\BigO{n\log{n}}$ interactions \whp/.
\end{proof}

\end{document}